\documentclass[12pt]{article}

\usepackage{amsmath}
\usepackage{amssymb}
\usepackage{amsthm}
\usepackage{array}
\usepackage{bm}
\usepackage[bf,footnotesize]{caption}
\usepackage{color}
\usepackage{enumitem}
\usepackage{eurosym}
\usepackage[top=2.4cm,left=2.7cm,right=2.7cm,bottom=3cm]{geometry}
\usepackage{graphicx}
\usepackage[sc]{mathpazo}
\usepackage{mathtools}
\usepackage{multirow}
\usepackage[numbers,sort&compress,square]{natbib}
\usepackage{setspace}
\usepackage{times}
\usepackage{titlesec}
\usepackage[normalem]{ulem}
\usepackage{verbatim}
\usepackage{xr}

\newcommand{\del}{\Delta}

\newcommand{\delhat}{\widehat{\Delta}}

\newcommand{\vx}{\mathbf{x}}
\newcommand{\vy}{\mathbf{y}}
\newcommand{\vz}{\mathbf{z}}
\newcommand{\vA}{\mathbf{A}}
\newcommand{\vB}{\mathbf{B}}

\newcommand{\MSS}{\circlearrowright\left(\vz,\lambda\right)}
\newcommand{\E}{\mathbb{E}}

\newcommand{\eq}[1]{Equation~\ref{eq:#1}}
\newcommand{\eqs}[2]{Eqs.~\ref{eq:#1}--\ref{eq:#2}}

\newcommand{\fig}[1]{Figure~\ref{fig:#1}}

\theoremstyle{definition}
\newtheorem{corollary}{Corollary}
\newtheorem{lemma}{Lemma}
\newtheorem{proposition}{Proposition}
\newtheorem{remark}{Remark}
\newtheorem*{fixation}{Fixation Axiom}

\title{\begin{center}
		\bfseries\singlespacing Strategy evolution on dynamic networks
\end{center}}
\author{\parbox[c]{16cm}{\onehalfspacing \normalsize \centering ~\\[-0.4cm]
		Qi Su$^{1,2,3,8,}\footnote{Correspondence: qisu@sjtu.edu.cn}$\qquad Alex McAvoy$^{4,5,8,}\footnote{Correspondence: amcavoy@unc.edu}$\qquad Joshua B. Plotkin$^{6,7}$\\ \quad\\ \footnotesize
		$^{1}$Department of Automation, Shanghai Jiao Tong University, Shanghai 200240, China \\
		$^{2}$Key Laboratory of System Control and Information Processing, Ministry of Education of China, Shanghai 200240, China \\
		$^{3}$Shanghai Engineering Research Center of Intelligent Control and Management, Shanghai 200240, China \\
		$^{4}$School of Data Science and Society, University of North Carolina at Chapel Hill, Chapel Hill, NC 27599, USA \\
		$^{5}$Department of Mathematics, University of North Carolina at Chapel Hill, Chapel Hill, NC 27599, USA \\
		$^{6}$Department of Biology, University of Pennsylvania, Philadelphia, PA 19104, USA \\
		$^{7}$Center for Mathematical Biology, University of Pennsylvania, Philadelphia, PA 19104, USA \\[0.35cm]
		$^{8}$These authors contributed equally to this work\\[0.2cm]}
	\date{}
}

\begin{document}
	
\maketitle

\begin{abstract}
\noindent 
Models of strategy evolution on static networks help us understand how population structure can promote the spread of traits like cooperation. One key mechanism is the formation of altruistic spatial clusters, where neighbors of a cooperative individual are likely to reciprocate, which protects prosocial traits from exploitation. But most real-world interactions are ephemeral and subject to exogenous restructuring, so that social networks change over time. Strategic behavior on dynamic networks is difficult to study, and much less is known about the resulting evolutionary dynamics. Here, we provide an analytical treatment of cooperation on dynamic networks, allowing for arbitrary spatial and temporal heterogeneity. We show that transitions among a large class of network structures can favor the spread of cooperation, even if each individual social network would inhibit cooperation when static. Furthermore, we show that spatial heterogeneity tends to inhibit cooperation, whereas temporal heterogeneity tends to promote it. Dynamic networks can have profound effects on the evolution of prosocial traits, even when individuals have no agency over network structures.
\end{abstract}

\section{Introduction}
The geographic locations of individuals, together with their social or physical connections, constrain interactions and shape behavioral evolution in a population. A network is a useful model of a population's structure, where nodes represent individuals and edges capture interactions. How network structure affects evolutionary dynamics has been extensively investigated over the last several decades, using techniques including computer simulations, mathematical analysis, and experimental studies with human subjects. A well-known and illustrative finding \cite{2006-Ohtsuki-p502-505} is that population structure can favor cooperation provided the ratio of the benefit from cooperative behavior, $b$, to its cost, $c$, exceeds the average number of neighbors, $d$. The mechanism underlying this cooperation-promoting effect is that spatial structure enables the formation of cooperative clusters of individuals, who have high payoffs and are capable of resisting invasion by defectors.

Most existing studies are based on a static network, where the duration and intensity of interactions remain unchanged throughout the evolutionary process. In contrast, empirical networks frequently vary over time \cite{2012-Holme-p97-125}. Representative examples include communication networks involving telephone calls or emails \cite{Vazquez-2007-prl,2007-Onnela-p7332-7336}; 
networks of physical proximity, where individuals encounter different people as they move through space \cite{Isella-2011-jtb,Mastrandrea2015}; and ecological networks that change with the seasons as organisms go through different phases of their life cycles \cite{Ulanowicz-2004-cbc,Bascompte-2007-AREES,Miele-2017-psos}. Temporal features can even reverse the evolutionary outcomes. For example, whether an idea or information diffuses throughout a society depends not only on the structure of the network guiding interactions but also on the timing of those interactions, as the coexistence of individuals with different active timing maximizes diffusion \cite{Akbarpour-2018-pnas}. In the context of epidemics, high concurrency (the number of neighbors of a node) leads to a lower epidemic threshold under susceptible-infected-susceptible dynamics, while low concurrency can suppress epidemics \cite{Onaga-2017-prl}.

Despite the attention that other dynamical processes have received on time-varying networks, the evolution of cooperation in this setting remains under-studied. One reason is that it seems unlikely, a priori, that dynamic changes in population structure will ever benefit cooperation. Since cooperators spread via clusters on a static network, we might expect that exogenous network transitions will tend to break up these clusters, leading to diminished reciprocity and exploitation by defectors \cite{2009-Kun-Biosystems,Fulker2021}. For example, switching between two random networks tends to impede cooperation relative to each separate network \cite{2009-Kun-Biosystems}. Another impediment to undertaking research in this area is the lack of mathematical tools for analyzing strategic interactions on dynamic networks. In static networks, mathematical approaches provide general conditions for how structure affects evolutionary dynamics \cite{taylor:Nature:2007,Allen2019}. They also allow for extensive, efficient numerical explorations into example networks, both artificial and empirical \cite{2017-Allen-p227-230}. Whether these approaches can be extended to dynamic networks remains unknown.

Endogenous network transitions often produce predictable results for the evolution of cooperation \cite{2006-Pacheco-p258103,2006-Santos-p1284-1291,2008-pacheco-jtb,2009-VanSegbroeck-p-,2010-Wu-p11187-11187,2011-Fehl-p546-551,2011-Rand-p19193-19198,Bravo2012,2012-Wang-p14363-14368,Bednarik2014,2014-Cardillo-p52825-52825,Harrell2018,2018-Akcay-p2692-2692}. For example, if cooperators can selectively seek out new connections with other cooperators (``cooperation begets friends'') and sever ties with defectors, then it is not surprising to find that these endogenous network changes favor the spread cooperation. But it is much less clear how exogenous transitions in network structure will affect the evolution of cooperation, and so this is the main focus of our study. There is also substantial evidence for the prevalence of exogenous network transitions in nature, ranging from weather fluctuations to human-induced changes to ecosystems \cite{wong:BE:2014}. The scope of models with dynamic networks is broad and can include environmental feedback and ecosystem engineering \cite{tilman:NC:2020}. And even when an organism has some agency over the structure of their environment, the behavioral trait of interest might be unrelated to these changes (e.g. movement between cities need not be tied to altruistic tendencies). Finally, exogenous network transitions that are not dependent on individual behavior provide the most natural point of comparison to static structures.
 
In this paper, we study the evolution of strategic behavior in a population whose structure of social interactions changes over time. At any point in time, the population structure is described by a network whose nodes represent individuals and edges represent interactions. Individuals may change their strategies over time, imitating neighbors who have higher payoffs; and the network of interactions itself may also change over time. The interaction network changes at random times, unrelated to the current composition of strategies in the population. We derive general mathematical results for when cooperative behavior is favored, which apply to any stochastic transition pattern among any number of networks, each with an arbitrary structure. Surprisingly, we find that in a large class of networks with community structure, stochastic transitions among networks can strongly promote cooperation, even though they tend to disrupt cooperative clusters in each network. In fact, even if each individual static network would disfavor cooperation, transitions among them can rescue cooperation. We conclude by analyzing spatial and temporal burstiness, which we show have opposite effects on the evolution of cooperation.

\section{Results}
\subsection{Model overview}
Our model consists of a finite population of size $N$, with individuals engaged in pairwise social interactions. The structure of the population varies over time, and at each discrete time it is represented by one of $L$ weighted networks, each with $N$ nodes. For network $\beta\in\left\{1,\dots ,L\right\}$, we let $w_{ij}^{\left[\beta\right]}$ denote the weight of the edge between nodes $i$ and $j$. We assume that all networks are undirected, meaning $w_{ij}^{\left[\beta\right]}=w_{ji}^{\left[\beta\right]}$ for all $i,j\in\left\{1,\dots ,N\right\}$ and $\beta\in\left\{1,\dots ,L\right\}$.
Supplementary Table 1 summarizes all quantities and associated notation used in our model formulation and analysis.

Each individual in the population can adopt one of two types, or strategies: ``cooperator'' ($C$) or ``defector'' ($D$). Individuals interact in pairwise donation games, with cooperators paying a cost $c$ to generate benefit $b$ for their co-player. Defectors pay no costs and generate no benefits. In each time step, everyone plays a donation game with each of their neighbors in the current network, $\beta$. We denote the state of the population by $\mathbf{x}$, where $x_{i}\in\left\{0,1\right\}$ indicates the type of individual $i$, with $0$ and $1$ representing types $D$ and $C$, respectively. The accumulated payoff to individual $i$ in network $\beta$ is then
\begin{equation}
u_{i}\left(\mathbf{x},\beta\right) = \sum_{j=1}^{N} w_{ij}^{\left[\beta\right]} \left(-cx_{i}+bx_{j}\right) . \label{eq:ui_db}
\end{equation}
In other words, individual $i$ receives a benefit $w_{ij}^{\left[\beta\right]}b$ from of each of its neighbors $j$ who are cooperators ($x_{j}=1$), and $i$ pays a cost $w_{ij}^{\left[\beta\right]}c$ to each $j$ if $i$ is itself a cooperator ($x_{i}=1$). An individual's accumulated payoff in network $\beta$ is transformed into fecundity, which represents $i$'s propensity to reproduce or, equivalently, to be imitated by another individual. The fecundity is given by $F_{i}\left(\mathbf{x},\beta\right) =1+\delta u_{i}\left(\mathbf{x},\beta\right)$, where $\delta$ is called the selection intensity, which we assume to be small ($\delta \ll 1$). This assumption, called ``weak selection," is common in the literature and it aims to capture scenarios in which the social trait ($C$ or $D$) has a small effect on reproductive success.

After all pairwise games are played in network $\beta$ and individuals accumulate payoffs, a random individual $i$ is selected uniformly from the population to update his or her strategy. This individual then imitates the type of a neighbor, $j$, with probability proportional to $j$'s fecundity. In other words, in network $\beta$, the probability that $i$ copies $j$'s type is
\begin{equation}
e_{ji}\left(\mathbf{x},\beta\right) = \frac{1}{N} \frac{F_{j}\left(\mathbf{x},\beta\right) w_{ji}^{\left[\beta\right]}}{\sum_{k=1}^{N}F_{k}\left(\mathbf{x},\beta\right) w_{ki}^{\left[\beta\right]}} . \label{eq:eji_db}
\end{equation}
Here, the factor of $1/N$ represents the probability that $i$ is chosen to update in the first place.

After each strategic update, the population structure itself then undergoes a transition step. The probability of moving from network $\beta$ to network $\gamma$ is independent of the strategic composition of the population, and it depends only on the current network state, $\beta$. The stochastic process governing these transitions is described by an $L\times L$ matrix $Q=\left(q_{\beta\gamma}\right)$, where $q_{\beta\gamma}$ is the probability of transitioning from network $\beta$ to network $\gamma$. Note that there may be (and we often assume) a positive chance that the network will remain unchanged at the transition stage, e.g. $q_{\beta \beta}>0$. The pairwise social interactions, strategic update, and network transition, which comprise a single time step, are depicted in Fig.~\ref{fig:model}.

\begin{figure}[ht]
	\centering
	\includegraphics[width=0.9\textwidth]{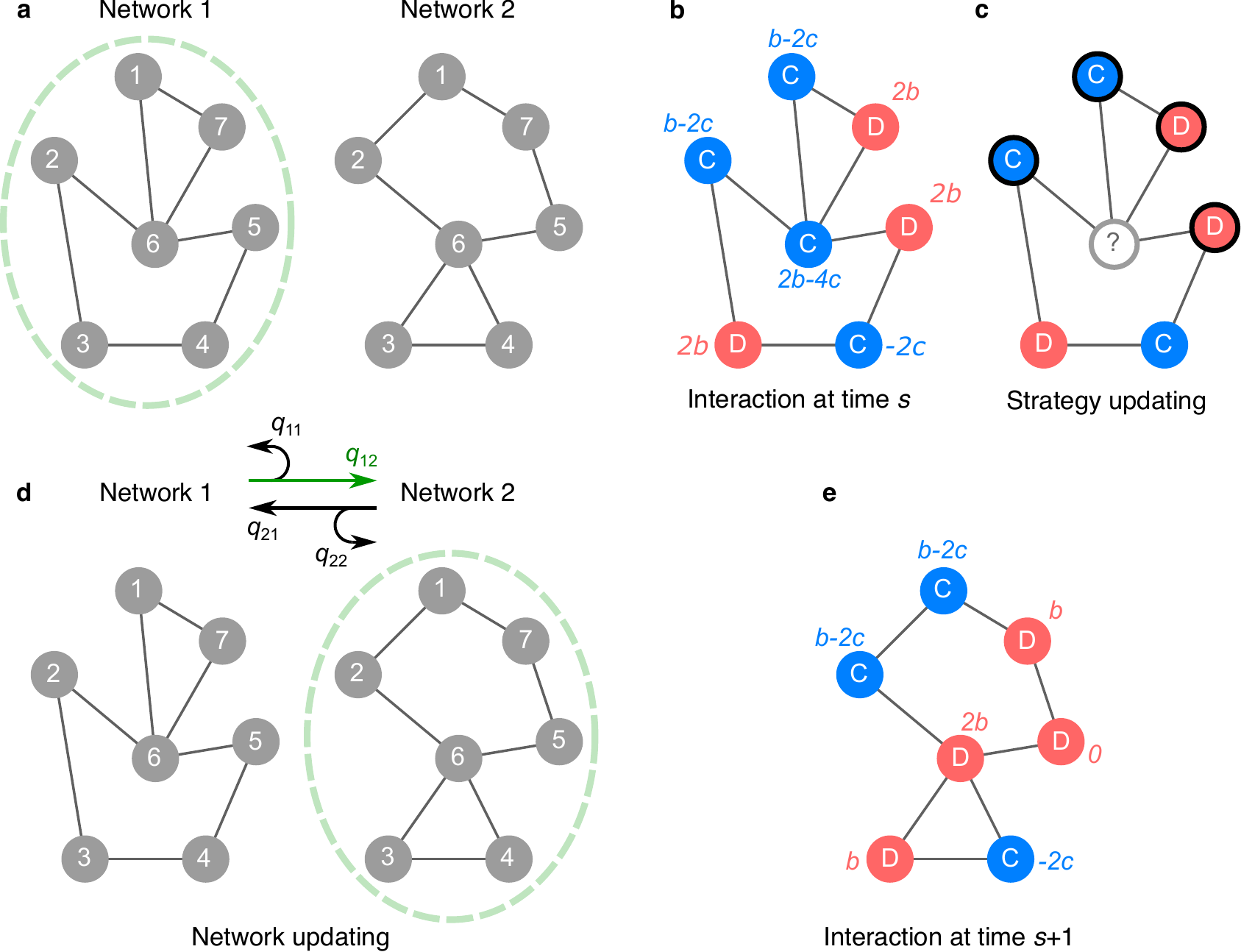}
	\caption {\textbf{Evolutionary games on dynamic networks.} \textbf{a}, The population structure at any time is described by a network, which may change from one time point to the next. (The figure illustrates an example with two possible networks.) \textbf{b}, Each individual (node) in the population adopts the strategy cooperate ($C$) or defect ($D$) in games played with neighbors. Each individual $i$ accumulates a total payoff $u_{i}$ across pairwise interactions with neighbors, which determines their fecundity $F_{i}=1+\delta u_{i}$. \textbf{c}, An individual (marked by ``?'') is selected uniformly at random to update its strategy, and all neighboring individuals, indicated by black circles, compete to be imitated by the focal node, with probability proportional to reproductive rates. \textbf{d}, After an individual updates its strategy, the population structure itself either changes (from network $1$ to network $2$ with probability $q_{12}$, or from network $2$ to network $1$ with probability $q_{21}$) or remains the same. \textbf{e}, Social interactions and strategy updates repeat on the population structure at the next time step, $s+1$.\label{fig:model}}
\end{figure}

\subsection{Selection condition for the evolution of cooperation}
Without mutation, the population must eventually reach a monomorphic strategic state in which all individuals have the same type, either cooperate or defect. The duration that the population spends in each network is proportional to the corresponding value in stationary distribution $\upsilon$, which is determined by the network transition matrix $Q$ (see Methods). We assume that a mutant appears in network $\beta$ with probability $\upsilon\left(\beta\right)$, and it is located at a node chosen uniformly at random. We let $\rho_{C}$ denote the probability that a single cooperator mutant eventually takes over a resident population of defectors. Likewise, we let $\rho_{D}$ be the probability that a single defector mutant takes over a resident population of cooperators. We use the condition $\rho_{C}>\rho_{D}$ to measure whether selection favors cooperation relative to defection \cite{2004-Nowak-p646-650}. 

We first derive a general result applicable to almost any transition pattern, $Q$, among any finite number of networks, each with an arbitrary spatial structure. This result combines several different quantities describing the dynamics under neutral drift ($\delta =0$), together with the payoffs for the game \cite{Allen2019,McAvoy2021}.

Let $p_{ij}^{\left[\beta\right]}\coloneqq w_{ij}^{\left[\beta\right]}/\sum_{k=1}^{N}w_{ik}^{\left[\beta\right]}$ be the one-step random-walk probability of moving from $i$ to $j$ on network $\beta$. This quantity can be interpreted as the probability that $i$ imitates the strategy of $j$ under neutral drift, conditioned on $i$ being chosen for an update. In other words, $p$ can be seen as defining an ancestral process, tracking replacement backwards in time under neutral drift.

The most fundamental neutral quantity is the reproductive value of individual $i$ in network $\beta$, which can be interpreted as the probability that a mutant introduced at node $i$ in network $\beta$ generates a lineage that eventually takes over the population. This quantity, denoted by $\pi_{i}^{\left[\beta\right]}$ is independent of the payoffs and thus independent of the particular mutant that arises in the population. The version of reproductive value that we use is a generalization of Fisher's classical notion \citep{fisher:OUP:1930,taylor:AN:1990} that also takes into account environmental changes. It can be calculated using \eq{reproductive_value_db} in Methods.

Another neutral quantity we use is related to coalescence times. Under neutral drift, we can look backward in time and ask how long it takes, on average, before two or more lineages meet at a common ancestor. Starting in network $\beta$, let $T^{\left[\beta\right]}$ be the expected number of steps to the most recent common ancestor of the entire population. If $\tau_{ij}^{\left[\beta\right]}$ is the expected time to the most recent common ancestor of $i$ and $j$, then the mean amount of time that $i$ and $j$ are identical by descent is $T^{\left[\beta\right]}-\tau_{ij}^{\left[\beta\right]}$. The pairwise times to a common ancestor, $\tau$, can be calculated using \eq{tau_db} in Methods.

In terms of the neutral quantities $\pi$, $\tau$, and $T$, the general condition for cooperation to be favored over defection under weak selection is given by
\begin{equation}
\begin{split}
\sum_{i,j=1}^{N} \sum_{\beta =1}^{L} &
\upsilon\left(\beta\right) \left(\sum_{\gamma =1}^{L}
q_{\beta\gamma}\pi_i^{\left[\gamma\right]}\right) p_{ij}^{\left[\beta\right]}\sum_{\ell =1}^{N}
\left(\substack{-\left(T^{\left[\beta\right]}-\tau_{jj}^{\left[\beta\right]}\right) w_{j\ell}^{\left[\beta\right]}c \\ +\left(T^{\left[\beta\right]}-\tau_{j\ell}^{\left[\beta\right]}\right) w_{\ell j}^{\left[\beta\right]}b}\right)  \\
&> \sum_{i,j,k=1}^{N} \sum_{\beta =1}^{L}
\upsilon\left(\beta\right) \left(\sum_{\gamma =1}^{L}
q_{\beta\gamma}\pi_i^{\left[\gamma\right]}\right) p_{ij}^{\left[\beta\right]}p_{ik}^{\left[\beta\right]}
\sum_{\ell =1}^{N}
\left(\substack{-\left(T^{\left[\beta\right]}-\tau_{jk}^{\left[\beta\right]}\right) w_{k\ell}^{\left[\beta\right]}c \\ +\left(T^{\left[\beta\right]}-\tau_{j\ell}^{\left[\beta\right]}\right) w_{\ell k}^{\left[\beta\right]} b}\right) . \label{eq:selection_condition}
\end{split}
\end{equation}
An outline of how to derive this condition is provided in Methods, while complete mathematical details are presented 
in Supplementary Sections 1-5.
Broadly speaking, what \eq{selection_condition} says is that an individual $i$ is chosen, a cooperator is placed at a neighbor $j$ of $i$, and another neighbor $k$ of $i$ is chosen to compare its (weighted) payoff with that of the cooperator. If $j$'s weighted payoff exceeds that of $k$, then selection favors the evolution of cooperation.

The condition above reflects a similar intuition behind the corresponding condition for static networks (see Allen et al.~\cite{2017-Allen-p227-230} or Fig.~1 of McAvoy \& Wakeley~\cite{mcavoy:PNAS:2022}), which corresponds to $L=1$ and is given by 
\begin{equation}
\sum_{i,j=1}^{N} 
\pi_i^{\left[1\right]}p_{ij}^{\left[1\right]}\sum_{\ell =1}^{N}
\left(\substack{-\left(T^{\left[1\right]}-\tau_{jj}^{\left[1\right]}\right) w_{j\ell}^{\left[1\right]}c \\ +\left(T^{\left[1\right]}-\tau_{j\ell}^{\left[1\right]}\right) w_{\ell j}^{\left[1\right]}b}\right)
>
\sum_{i,j,k=1}^{N} \pi_i^{\left[1\right]}
p_{ij}^{\left[1\right]}p_{ik}^{\left[1\right]}
\sum_{\ell =1}^{N}
\left(\substack{-\left(T^{\left[1\right]}-\tau_{jk}^{\left[1\right]}\right) w_{k\ell}^{\left[1\right]}c \\ +\left(T^{\left[1\right]}-\tau_{j\ell}^{\left[1\right]}\right) w_{\ell k}^{\left[1\right]} b}\right). 
\end{equation}
Compared to a static network, there are a few notable effects of network transitions in \eq{selection_condition}. The first effect is that the network $\beta$ is chosen with probability $\upsilon\left(\beta\right)$, where $\upsilon$ is the stationary distribution of the structure-transition chain defined by $Q$. Moreover, whereas individual $i$ is chosen with probability based on reproductive value $\pi_{i}$ on a static network, here $i$ is chosen based on reproductive value in the \emph{next} network following imitation, $\sum_{\gamma =1}^{L}q_{\beta\gamma}\pi_{i}^{\left[\gamma\right]}$. The reason for this is natural, because once an individual replaces $i$ in network $\beta$, the network immediately transitions to network $\gamma$, and so the resulting reproductive value of $i$ must be understood within the context of $\gamma$. Once $\beta$ and $i$ are chosen, the probabilities of choosing neighbors $j$ and $k$ are $p_{ij}^{\left[\beta\right]}$ and $p_{ik}^{\left[\beta\right]}$, respectively. Moreover, if $j$ is a cooperator, then individual $k$ is also a cooperator for $T^{\left[\beta\right]}-\tau_{jk}^{\left[\beta\right]}$ time steps, and during each such step $k$ pays $cw_{k\ell}^{\left[\beta\right]}$ to provide $\ell$ with a benefit of $bw_{k\ell}^{\left[\beta\right]}$. This property accounts for the weighting of benefits and costs in \eq{selection_condition}. Note that the term $T^{\left[\beta\right]}$ cancels out in \eq{selection_condition}, and so although this quantity is helpful for gathering intuition, it is not strictly needed to evaluate whether cooperators are favored by selection.

Given the vast number of networks with $N$ nodes, as well as the vast space of possible transitions among them, we focus most of our analysis on transitions between a pair of networks (i.e. $L=2$). For a given network transition matrix $Q$, the value $1/q_{12}$ (resp. $1/q_{21}$) gives the expected time during which the population remains in network $1$ (resp. network $2$) before transitioning to network $2$ (resp. network $1$). We denote $1/q_{12}$ and $1/q_{21}$ by $t_{1}N$ and $t_{2}N$, respectively, so that $t_{1}$ and $t_{2}$ correspond to the expected number of times each individual updates prior to a transition to a different network. $t_{1}$ and $t_{2}$ are in the units of generations. Small values of $t_{1}$ and $t_{2}$ correspond to frequent changes in the population structure. Sufficiently large values of $t_{1}$ and $t_{2}$ indicate that the population structure is nearly fixed, so that the population will reach an absorbing strategic state (all $C$ or all $D$) before the network transitions to a different state. The regime $t_{1}=1$ (resp. $t_{2}=1$) means that, on average, each individual updates their strategy once in network $1$ (resp. network $2$) before the network structure changes.
In the following, we focus on cases with $t_1=t_2=t$.

\subsection{Dynamic networks with dense and sparse communities}
We begin by studying dynamic transitions between a pair of networks where each network is comprised of two communities. One community is a star graph, which is sparse, and the other community is a complete graph, which is dense. In each network, the two communities are connected by a single edge. When the population transitions from one network to another, the star community becomes the complete community and \emph{vice versa} (see \fig{star-complete}\textbf{a}). This kind of dynamic network models a situation in which a portion of the population is densely connected while the remainder of the population is connected to only a single node; and which portion is dense versus sparse changes over time, as the state transitions between the two networks.

\begin{figure}
\centering
\includegraphics[width=1\textwidth]{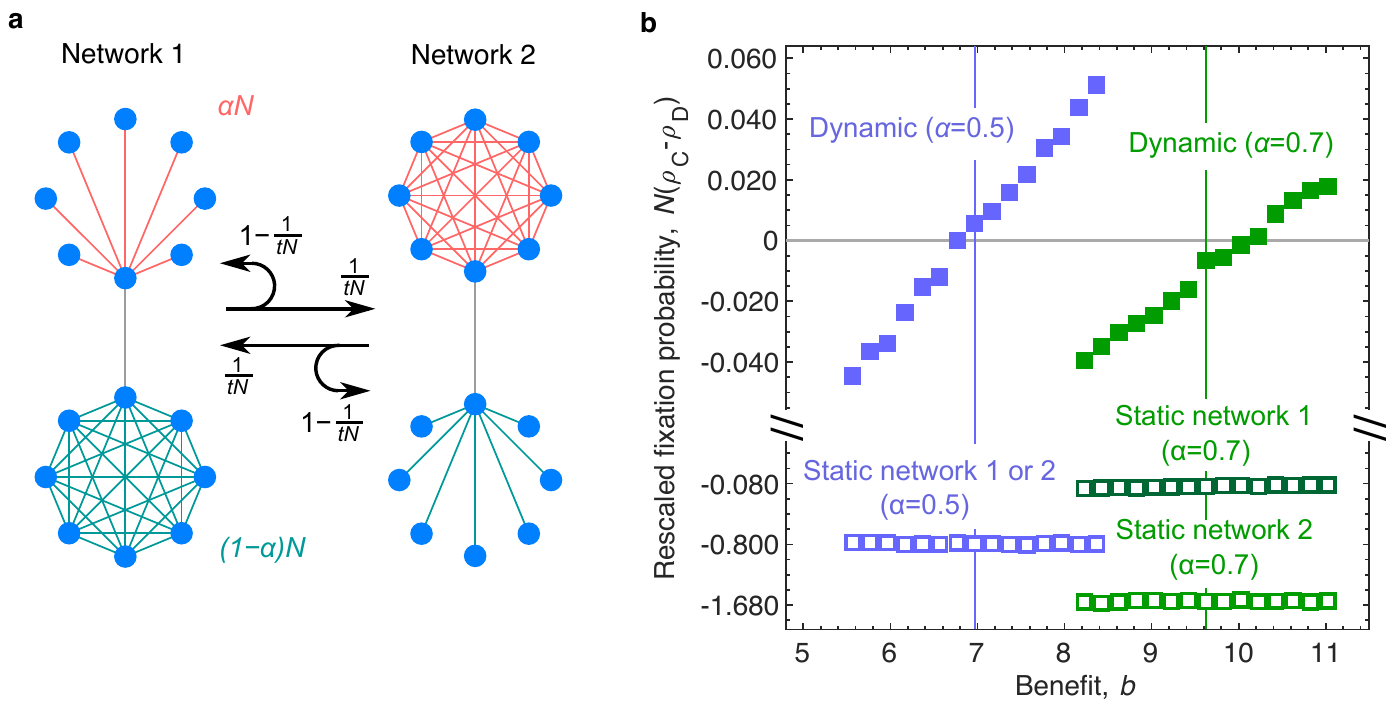}
\caption{\textbf{Transitions between networks that contain dense and sparse communities.} We consider dynamic transitions between two networks, each of which is comprised of two communities containing $aN$ and $\left(1-a\right) N$ nodes, respectively. \textbf{a}, Each network has a star graph comprising one community and a complete graph comprising the other community, with a single edge connecting the two communities. When network $1$ transitions to network $2$, the star community becomes the complete community and \emph{vice versa}. \textbf{b}, The fixation probability of cooperation versus defection, $\rho_{C}-\rho_{D}$, as a function of the benefit, $b$, in the donation game. Selection favors cooperation over defection if $\rho_{C}-\rho_{D}$ exceeds the horizontal line, i.e., $\rho_{C}>\rho_{D}$. Dots indicate Monte Carlo simulations on dynamic networks (solid dots) and on a static network (open dots). The vertical lines correspond to analytical calculations of the critical benefit-to-cost ratio $\left(b/c\right)^{\ast}$ on dynamic networks, above which cooperation is favored. Cooperation is always disfavored in both static network $1$ and static network $2$ (separately), but dynamic transitions between these networks can favor cooperation. Here, we show two examples with different community sizes, $a=0.5$ (blue) and $a=0.7$ (green). The beneficial effect of structure transitions is strongest when communities have equal size ($a=0.5$; see Supplementary Figure 1). Parameter values: $N=40$, $t=1$, and $c=1.0$. Fixation probabilities are computed across an ensemble of $10^{7}$ runs with selection intensity $\delta =0.002$.\label{fig:star-complete}
}
\end{figure}

When the population evolves on either network $1$ or network $2$ alone, the fixation probability of cooperators is always lower than that of defectors, i.e. $\rho_{C}<\rho_{D}$, meaning that cooperation is disfavored by selection regardless of the benefit-to-cost ratio $b/c$ (\fig{star-complete}\textbf{b}). Nonetheless, when the population transitions dynamically between networks 1 and 2, cooperation is favored provided the benefit-to-cost ratio $b/c$ exceeds the critical value $\left(b/c\right)^{\ast}\approx 7$. As a result, we see that dynamic population structures can favor cooperation, even when all networks involved would each individually suppress cooperation were they static.

Dynamic population structure facilities cooperation across a wide range of population sizes for the pair of networks shown in \fig{star-complete}\textbf{a}. When $t_1=t_2=1$, which means that individuals each update their strategy once, on average, before the network changes, cooperation can be favored by selection regardless of network size, $N$ (\fig{size-duration}\textbf{a}). By contrast, if the network is static, then cooperation is favored only when the population size is very small ($N<17$)--and, even then, only if the benefit-to-cost ratio is large. For larger population sizes, $N\geqslant 17$, the critical benefit-to-cost ratio is negative on a static network, $\left(b/c\right)^{\ast}<0$, which means that selection actually favors the evolution of spite, a behavior in which individuals pay a cost $c$ to decrease the fitness of their opponent by $b$. For this static network, we can show that $\left(b/c\right)^{\ast}\approx -N/2$ in large populations (see Methods), whereas the critical ratio for cooperation approaches a constant positive value, $\left(b/c\right)^{\ast}=7$,  for large $N$ in these dynamic networks (see \eq{sparse_dense_dynamic} for the case of $a=1/2$ and $t=1$). And so the effects of dynamic population structures can be dramatic, capable of converting a spiteful outcome into a cooperative one, and they persist across a wide range of population sizes.

\begin{figure}
\centering
\includegraphics[width=1\textwidth]{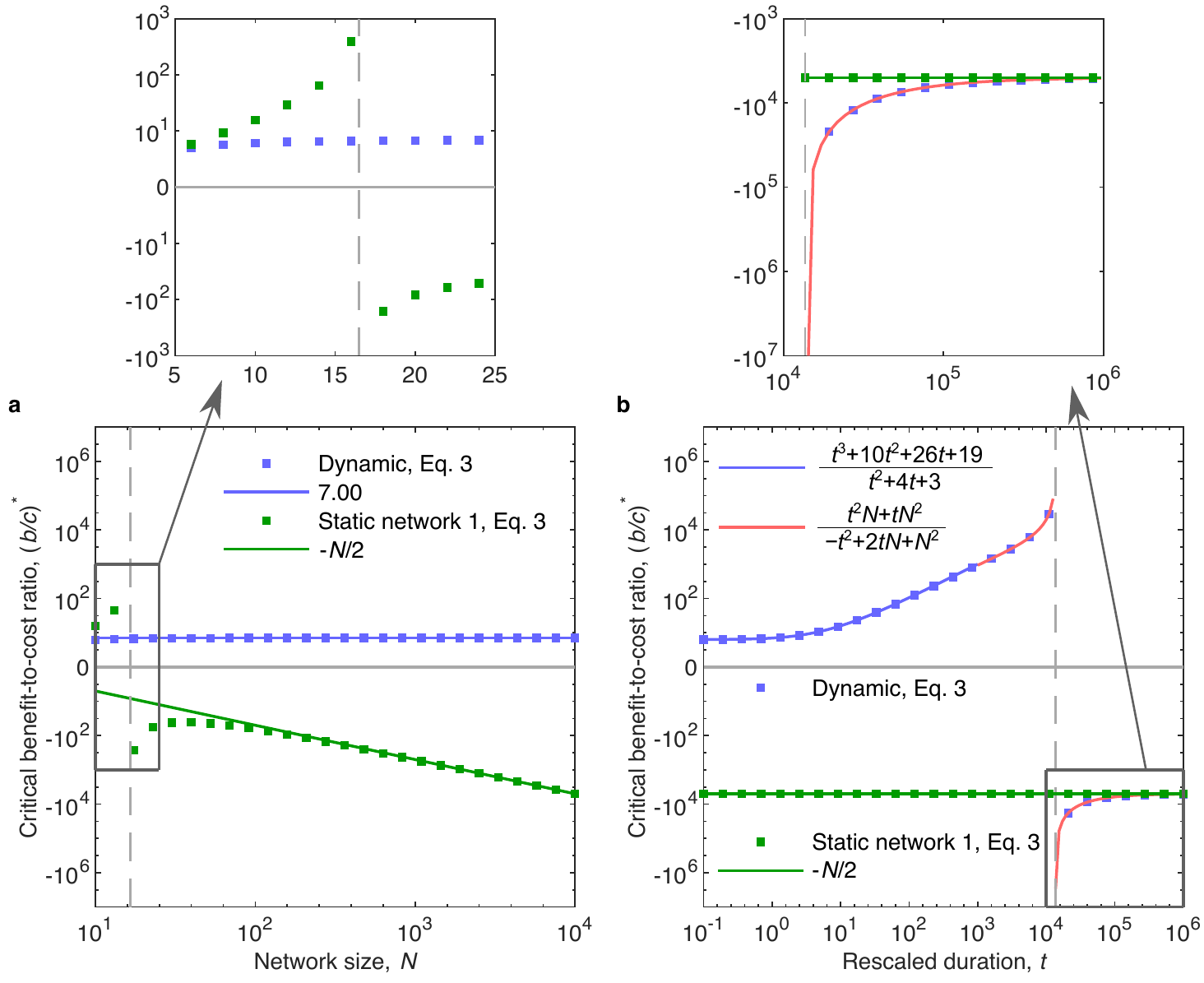}
\caption{\textbf{Dynamic structures facilitate cooperation for a broad range of population sizes and network transition rates.} We consider transitions between the two networks shown in \fig{star-complete}\textbf{a}, each composed of a sparse community and a dense community. \textbf{a}, The critical benefit-to-cost ratio required to favor cooperation as a function of population size, $N$, for community size $a=0.5$ and mean duration $t=1$. Dynamic networks can favor cooperation for any population size, $N$, provided $b/c>7$. In contrast, the corresponding static networks favor cooperation only in small populations ($N<17$), and they favor the evolution of spite ($\left(b/c\right)^{\ast}<0$) in larger populations. Dots show exact analytical computations for finite $N$ (\eq{selection_condition}), and lines show analytical approximations for large $N$. \textbf{b}, The critical benefit-to-cost ratio as a function of the mean duration between network transitions, $t$, for $a=0.5$ and $N=10{,}000$. Whereas a static network always disfavors cooperation, dynamic networks can favor cooperation provided they do not transition too slowly ($t<\left(\sqrt{2}+1\right) N$). Dots show exact analytical calculations for arbitrary $t$; the blue line shows an analytical approximation in the regime $t\ll N$; and the red line shows an analytical approximation in the regime $t=O\left(N\right)$.\label{fig:size-duration}}
\end{figure}

Dynamic networks also facilitate cooperation across a wide range of structural transition rates. For a sufficiently large population size, $N$, on a single static network of the type shown in \fig{star-complete}\textbf{a}, the critical benefit-to-cost ratio is negative ($\left(b/c\right)^{\ast}\approx -N/2$), which means that selection favors the evolution of spite. By contrast, dynamic transitions between networks 1 and 2 can favor cooperation, especially when they occur rapidly (\fig{size-duration}\textbf{b}). When the transition rate is very slow -- in particular, when $t$ exceeds $\left(\sqrt{2}+1\right) N$ -- the population stays in one network for so long that the evolutionary dynamics are similar to those of a static network, and the critical benefit-to-cost ratio becomes negative (\fig{size-duration}\textbf{b}). In the limit of the transition rate approaching zero ($t\rightarrow \infty$), the ``dynamic'' network is actually static and our dynamic calculations agree with those of a static network.

\subsection{How dynamic structures can facilitate cooperation}
To further understand how dynamic structures can favor cooperation more than their static counterparts, we inspect evolutionary trajectories on the dense-sparse graph of \fig{star-complete}\textbf{a}. When the network is static, the process is depicted in \fig{intuition}\textbf{a}. Starting from a specific configuration of cooperators in both hubs and two leaf nodes, cooperation will initially tend to spread in the star community while shrinking in the complete community. After cooperation fixes within the star community, selection strongly suppresses further spread to the complete community because the node connected to the star community is exploited by multiple defectors. If ever a defector manages to diffuse to the hub of the star community, however, defection will then rapidly spread within the star and ultimately fix in the entire network.

\begin{figure}
\centering
\includegraphics[width=0.7\textwidth]{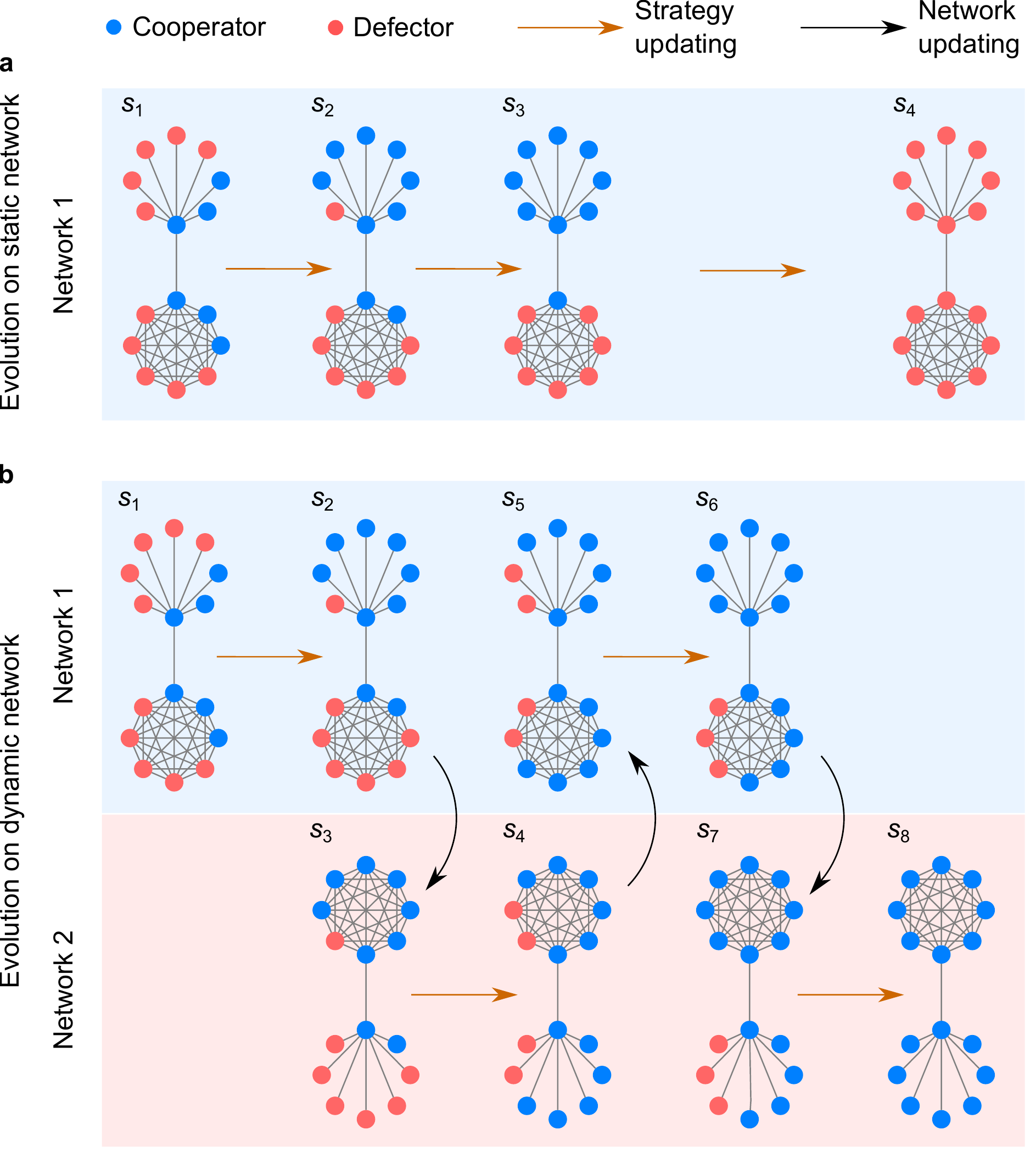}
\caption{\textbf{Intuition for how dynamic structures can facilitate cooperation.} Starting from a configuration in which the hub and two leaf nodes are cooperators (time point $s_{1}$ in \textbf{a} and \textbf{b}), we illustrate how cooperation can be favored in dynamic structures even when it is inhibited in each static structure. Initially, cooperators are expected to spread in the star community and shrink in the complete community, and the rate of spreading exceeds that of shrinking. \textbf{a}, The evolutionary process on a static network. Cooperators rapidly take over the star community and nearly die out in the complete community ($s_{1} \rightarrow s_{3}$). The system tends to stay in this state until defectors spread throughout the star community ($n_{s}$). \textbf{b}, The evolutionary process with network transitions. Initially, cooperators spread in the star community and shrink in the complete community ($s_{1}\rightarrow s_{2}$). However, when the network changes, the star community transitions to the complete community and \emph{vice versa} ($s_{2}\rightarrow s_{3}$). This transition is followed by the rapid spread of cooperators in the star community and (relatively slower) shrinking of cooperators in the complete community ($s_{3}\rightarrow s_{4}$). From $s_{1}$ to $s_{5}$, the frequency of cooperators increases in both communities so that, under dynamic structure transitions, cooperators tend to fix in both communities ($s_{8}$).\label{fig:intuition}}
\end{figure}

By contrast, if the population undergoes structural transitions between networks (e.g. $s_{2} \rightarrow s_{3}$ in \fig{intuition}\textbf{b}), the star community of network $1$ will transition into the complete community of network $2$, which promotes the exploitation of cooperators and allows defectors to spread ($s_{3} \rightarrow s_{4}$). Meanwhile, the complete community of network $1$ transitions into the star community of network $2$, which stimulates the expansion of cooperators. The rate of cooperator expansion in one community exceeds their exploitation in the other community (see the qualitative analysis in Supplementary Figure 2), so that, overall, network transitions facilitate cooperation.

\subsection{Other dynamic structures}
The examples of dynamic structure considered so far may seem highly specialized because the networks each contain two stylized communities with a single edge between them. But we find similar results on networks with many communities and with more complicated connections between them. In \fig{other_structures}\textbf{a,b}, we analyze networks comprised of multiple star and complete communities, connected by either hub nodes or by leaf nodes. In both cases, we again find that dynamic transitions between networks reduce the critical benefit-to-cost ratio for the evolution of cooperation, compared to any single static network. This effect is increasingly strong as the network size grows (see Supplementary Figure 3). For the networks in \fig{other_structures}\textbf{a} with $N=1{,}200$, for example, the critical benefit-to-cost ratio to favor cooperation is $\left(b/c\right)^{\ast}\approx 188.1$ when the network is static, which is reduced to $\left(b/c\right)^{\ast}\approx 3.49$ when the network is dynamic.

\begin{figure}
\centering
\includegraphics[width=1\textwidth]{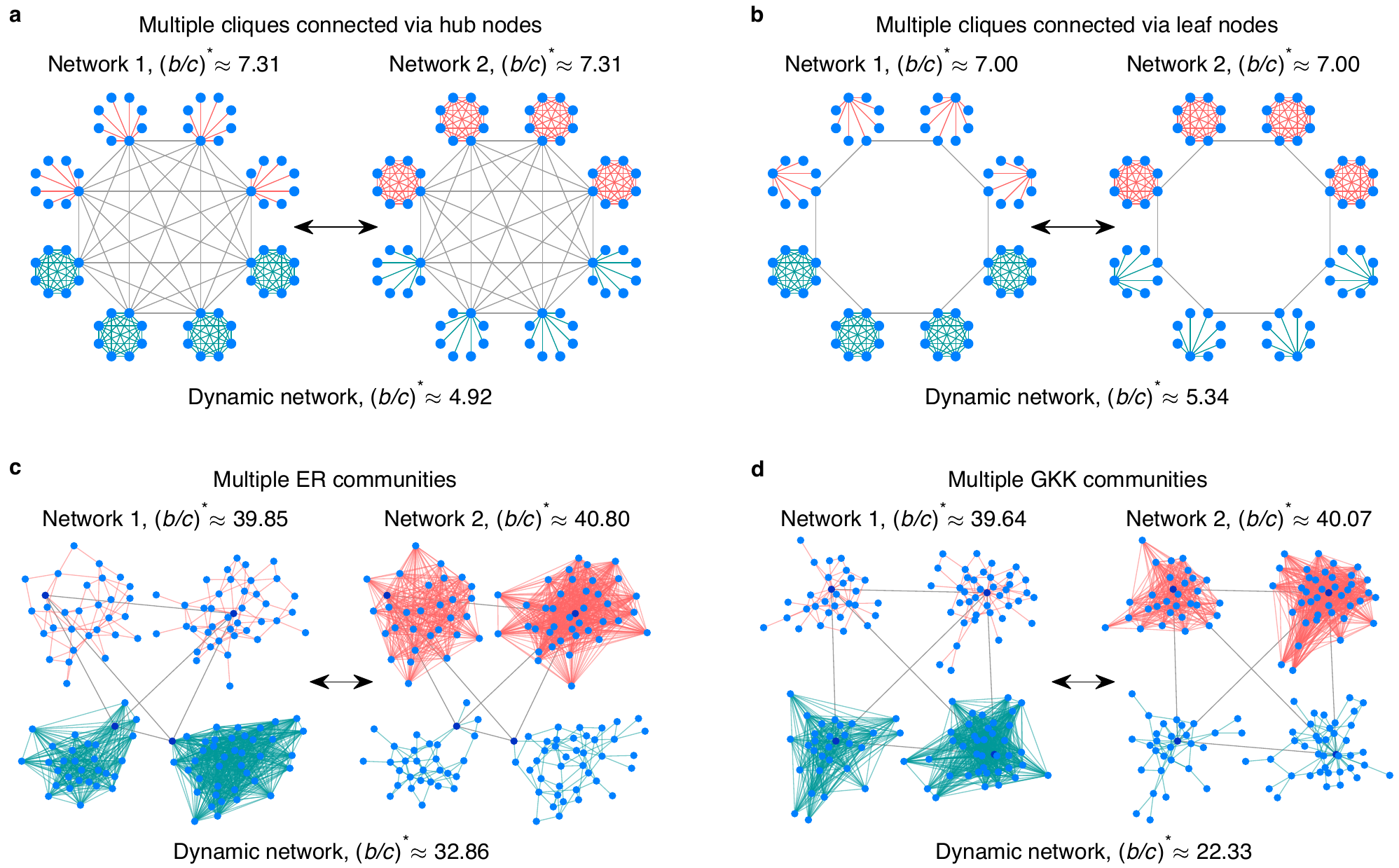}
\caption{\textbf{Evolution of cooperation on diverse dynamic structures.} \textbf{a}, Each individual network comprises four star communities and four complete communities, where each star community in one network corresponds to a complete community in the other network, and community hubs are fully connected to each other. \textbf{b} is similar to \textbf{a}, but communities are now sparsely connected via leaf nodes. Network transitions facilitate cooperation compared to a static structure. \textbf{c}, Each individual network comprises two sparse and two dense communities of Erd\"{o}s-R\'{e}nyi (ER) random networks \cite{1960-Erdoes-p17-61}, with communities connected by random nodes. \textbf{d}, Each individual network comprises two sparse and two dense communities of Goh-Kahng-Kim scale-free networks (GKK) \cite{2001-Goh-p278701-278701} with exponent $2.5$, with communities connected by nodes of the highest degree. In all these examples, network transitions reduce the benefit-to-cost ratio $\left(b/c\right)^{\ast}$ required for cooperation compared to each static network. Parameters: mean duration $t=1$ and population size $N=64$ for \textbf{a} and \textbf{b}. For panels \textbf{c} and \textbf{d}, in network $1$, the two sparse communities have $30$ nodes and average degree $4$, and the two dense communities have $40$ nodes and average degree $30$; in network $2$, the two sparse communities have $40$ nodes and average degree $4$, and the two dense communities have $30$ nodes and average degree $20$.\label{fig:other_structures}}
\end{figure}

In addition to networks comprised of star and complete communities, we also investigated networks with communities defined by various types of random graphs, such as Erd\"{o}s-R\'{e}nyi and scale-free networks. In the former case, node degrees within a community do not vary substantially, while the latter exhibits large variation in degree. 
Here we emphasize that when dynamic transitions occur between two random networks or two random scale-free networks that lack community structure, then cooperation is impeded (relative to a static random or a static scale-free network) -- for the intuitive reason that random switches disrupt cooperative clusters \cite{2006-Ohtsuki-p502-505}. 
But in the multi-community random networks, dynamic transitions between random networks can promote cooperation, compared to each static network (\fig{other_structures}\textbf{c,d}).

In all examples of dynamic networks considered thus far, transitions between networks involve dense regions of a network swapping with sparse regions. Regardless of the exact structure of the communities (star-like, complete, random, or scale-free communities), this general feature of structural transitions conforms to the underlying intuition for why dynamic networks can facilitate cooperation in networks with communities (\fig{intuition}). Dynamic structures can still facilitate cooperation even when networks differ in only a small fraction of connections, although the strength of the effect is weakened. Furthermore, these effects also persist (and can be quite strong) when populations transition between three or more network structures. We give illustrations in Supplementary Figure 4.

\subsection{The probability and time to fixation of cooperation}
We have studied dynamic structures by comparing the fixation probability of a cooperator to that of a defector, and by calculating the critical benefit-to-cost ratio $\left(b/c\right)^{\ast}$ that ensures $\rho_{C}>\rho_{D}$ in a large class of dynamic networks. We can also study the fixation probability $\rho_{C}$ in absolute terms. We find that a dynamic population structure increases the fixation probability of cooperators, making them more likely to overtake the population, compared to a static network. Dynamic population structures also tend to decrease the duration before one type or another fixes (see Supplementary Figure 5), as well as shorten the mean conditional time until cooperators fix. The underlying intuition for these results is evident in \fig{intuition}: on a static network, the population will tend to be stuck at stage $s_3$ for a long time, before defectors eventually diffuse to the sparse community; whereas on dynamic networks, cooperators spread rapidly by selection in both communities. Thus, dynamic networks increase the likelihood that cooperators sweep the population as well as the rate at which they do so.

\subsection{Spatial and temporal burstiness}
We can adapt our method of analysis to study the effects of spatial and temporal burstiness. For dynamically changing networks, spatial burstiness arises when there is temporal variation in the density of network edges (node degree), whereas temporal burstiness arises when there are periods of rapidly changing network structures along with periods in which structures change more slowly. Empirical networks of both human and non-human (e.g. honeybee) interactions are known to exhibit both spatial and temporal burstiness \cite{Akbarpour-2018-pnas,gernat:PNAS:2018}, but the effects of these two forms of over-dispersion for behavior remains an active area of current research.

To study spatial burstiness, we consider the following minimal model of dynamically varying homogenous networks that differ in their average node degree. We construct a pair of networks as follows (see \fig{spatial_temporal}\textbf{a}): \emph{(i)} we first generate a single network with $N$ nodes and $E$ edges drawn from one of several classical families of networks (e.g. Erdos-Reyni random networks \cite{1960-Erdoes-p17-61}, Watts-Strogatz small-world networks \cite{1998-Watts-p440-442}, Barab\'{a}si-Albert scale free networks \cite{1999-Barabasi-p509-512}, etc.); \emph{(ii)} we decompose this network into two networks, by randomly selecting a fraction $\varepsilon\in\left[0,1/2\right]$ of the edges for network $1$ and using the remaining $\left(1-\varepsilon\right) E$ edges for network $2$. If $\varepsilon =1/2$ then the resulting networks 1 and 2 have the same density of interactions, and there is no spatial burstiness. For all other values of $\varepsilon\neq 1/2$, the network exhibits spatial burstiness, and we study a simple stochastic transition pattern between these networks, with $t_{1}=t_{2}=1$ so that each individual updates his strategy once, on average, before the network switches.

\begin{figure}
\centering
\includegraphics[width=1\textwidth]{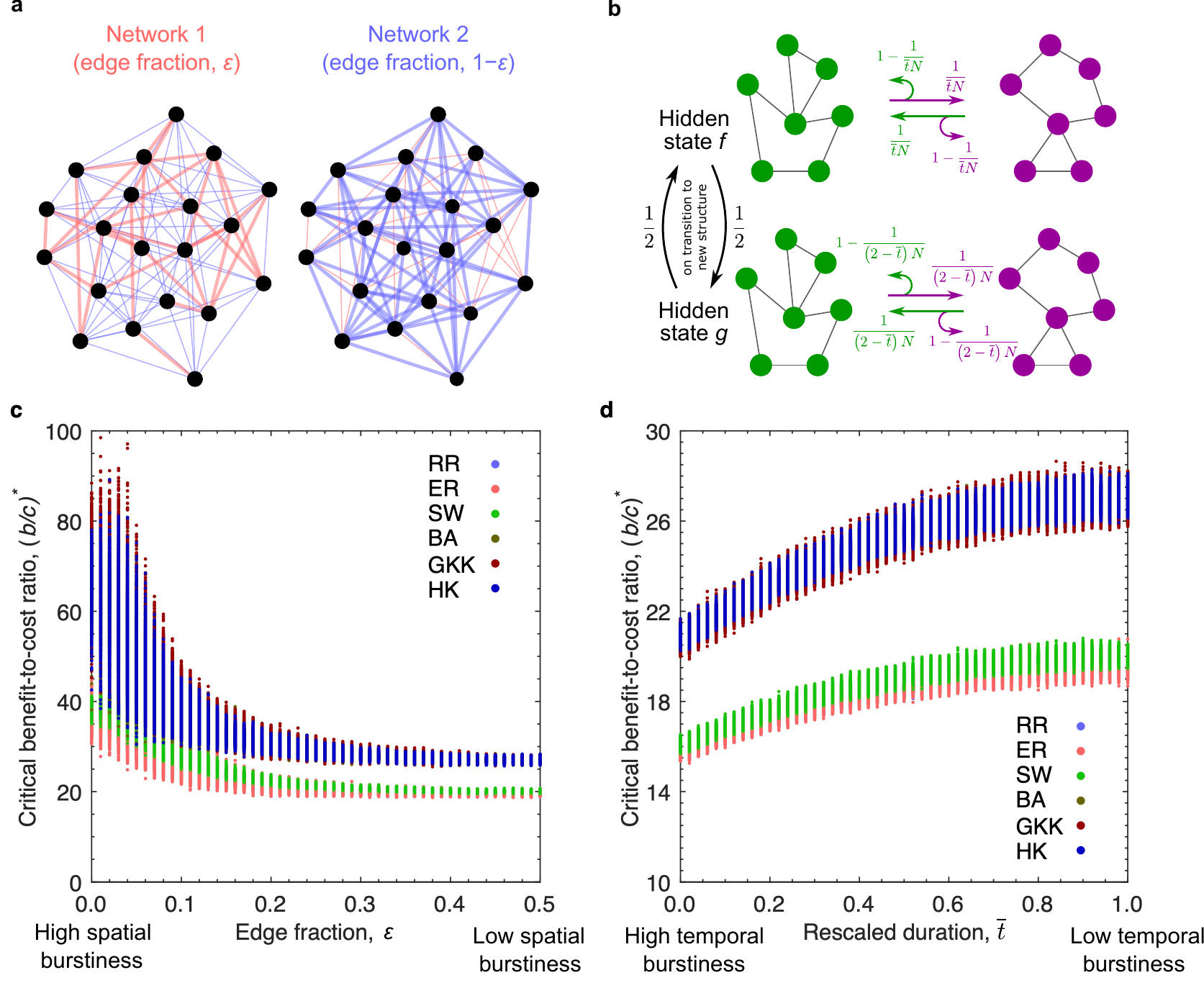}
\caption{\textbf{Effects of spatial and temporal burstiness on cooperation.} We consider transitions between two networks, with either \textbf{a}, spatial burstiness (different edge densities) or \textbf{b}, temporal burstiness (periods of both rapid and slow transitions). \textbf{c}, The critical benefit-to-cost ratio $\left(b/c\right)^{\ast}$ as a function of spatial heterogeneity, $\varepsilon$. When the two networks have the same edge density, $\varepsilon =0.5$, cooperation is most readily favored. When the networks that differ in their edge densities ($\varepsilon \ll 0.5$), much larger values of the benefit-to-cost ratio, $b/c$, are required to support cooperation. \textbf{d} The critical benefit-to-cost ratio $\left(b/c\right)^{\ast}$ required to favor cooperation as a function of temporal heterogeneity, $\bar{t}$. The case $\bar{t}=1$ means that networks transition at the same rate, regardless of the hidden state. When $\bar{t}<1$, the networks transition more rapidly in state $f$ than in state $g$, so that there is temporal burstiness. Critical benefit-to-cost ratios under spatial and temporal burtiness are shown for six classes of networks: random regular networks (RR), Erd\"{o}s-R\'{e}nyi networks (ER) \cite{1960-Erdoes-p17-61}, Watts-Strogatz small-world networks (SW) \cite{1998-Watts-p440-442} with rewiring probability $0.1$, Barab\'{a}si-Albert scale-free networks (BA) \cite{1999-Barabasi-p509-512}, Goh-Kahng-Kim scale-free networks (GKK) \cite{2001-Goh-p278701-278701} with exponent $2.5$, and Holme–Kim scale-free networks (HK) \cite{Holme2002} with triad formation probability $0.1$. For each such class, we generate $1{,}000$ networks, each with $100$ nodes and average degree $20$. We take $\bar{t}=1$ in \textbf{c} and $\varepsilon =0.5$ in \textbf{d}.\label{fig:spatial_temporal}}
\end{figure}

We find that spatial burstiness tends to inhibit the evolution of cooperation, whereas spatial regularity (equal network densities) is more beneficial for cooperation (\fig{spatial_temporal}\textbf{c}). In particular, regardless of the class of network from which networks 1 and 2 are derived, the critical ratio $\left(b/c\right)^{\ast}$ required to favor cooperation is substantially increased (roughly by a factor of two) in the regime $\varepsilon \rightarrow 0$ compared to the spatially homogeneous regime $\varepsilon =1/2$.
The dynamic networks with spatial burstiness shown in \fig{spatial_temporal}\textbf{a,c}, which impede cooperation, differ fundamentally from the dynamic networks studied in Figures \ref{fig:star-complete}-\ref{fig:other_structures}, which facilitate cooperation.  The networks in Figures \ref{fig:star-complete}-\ref{fig:other_structures} each contain multiple communities with different local edge densities; and the densities switch when the network changes. By contrast, the dynamic networks shown in \fig{spatial_temporal}\textbf{a,c} (spatial burstiness) contain no communities but rather they have a uni-modal edge degree distribution in every state: these homogeneous networks simply switch between high-degree and low-degree distributions, which tends to impede cooperation. This result aligns with previous simulation studies of random network rewiring \cite{2009-Kun-Biosystems}. The intuition for this result is simple: cooperative partnerships may arise in the sparse homogeneous state, but transitions to a dense homogeneous network allow defectors to invade cooperative clusters. 

We also study the effects of temporal burstiness, in which case networks 1 and 2 are chosen to have the same edge density ($\varepsilon =1/2$), but there are periods of rapid transitions between the two networks, punctuated by periods of slow transitions. To construct this scenario, instead of having a single transition matrix, $Q$, we consider two such matrices, $Q^{f}$ and $Q^{g}$, corresponding to fast and slow epochs. At any time, the population is either in hidden state $f$, so that network transitions occur according to $Q^{f}$, or alternatively in hidden state $g$, so that network transitions occur according to $Q^{g}$. Whenever the population transitions to a new network, the hidden state is drawn uniformly-at-random from $\left\{f,g\right\}$ (see \fig{spatial_temporal}{\bf{b}}). (Note that the hidden state $g$ or $f$ is re-sampled only when the network changes, from 1 to 2 or from 2 to 1.)

The speed of network transitions in each hidden state, $g$ and $f$, is governed by a parameter $\bar{t}\in\left[0,1\right]$, so that transitions are fast in state $f$ and slow in state $g$. When the population enters state $f$, the expected duration before a network transition is small, namely $\bar{t}N$. Whereas when the population enters state $g$ the expected duration of the current network is longer, $\left(2-\bar{t}\right) N$ (see \fig{spatial_temporal}\textbf{b}). The case $\bar{t}=1$ means that the current network has the same expected duration, regardless of the hidden state, and there is no temporal burstiness. When $\bar{t}<1$, the networks transition more quickly in state $f$ than they do in state $g$. Regardless of the value of $\bar{t}$, however, the total accumulated time spent in network $1$ is the same as in network $2$, throughout the evolutionary process.

Temporal burstiness tends to facilitate cooperation in networks lacking community structure (\fig{spatial_temporal}\textbf{d}). In particular, the critical benefit-to-cost required to favor cooperation is largest when temporal burstiness is absent ($\bar{t}=1$), and it is reduced (typically by $~20$\%) when temporal burstiness is large ($\bar{t}=0$). 
In the case of high temporal burstiness, the system experiences an epoch of time with frequent changes in network states, followed by an epoch with long periods of stability within a network. The epoch with frequent network changes is particularly detrimental to the survival of cooperator clusters. However, the extended periods of stability promote the formation of cooperator clusters. Even when two networks have the same edge density ($\varepsilon=1/2$) and the accumulated time spent on each network is the same, large 
temporal burstiness allows the establishment of large cooperator clusters in epochs with slow transitions, and this promotes the survival of cooperators, in stark contrast to our findings for spatial burstiness.

\section{Discussion}
Many real-world interactions are ephemeral, and the entire network of social interactions may be subject to exogenous changes. Seasonal changes in a species' environment, for example, can lead to active and dormant periods, as can diurnal cycles. Such periodic transitions are widely used to model temporal networks \cite{2005-Lieberman-p312-316,Pan-2011-pre,2015-Holme-p234-234,2015-Valdano-p21005-21005}. Stochastic transitions in social structures can arise from the effects of weather, animal migration and movement, and role reversal \cite{povinelli:AB:1992}. Motivated by the ubiquity of structural variation in nature, we provide a treatment of dynamic social networks that allows for arbitrary stochastic transitions between structures, with arbitrary networks within each time step.

Our main mathematical result (\eq{selection_condition}) shows when cooperation will evolve on dynamic networks, under weak selection. The population structure in every time step need not be connected; all that we require is that the population satisfy a coherence condition so that it does not become fragmented into multiple sub-populations (see Supplementary Section 1). In addition to probabilistic transitions, our analysis extends to deterministic and periodic network transitions (see Equation 33 in Supplementary Information). Our analysis also applies to other scenarios for changing structures, such as when the direction of public goods or information flow changes over time\cite{su-2022-pnas}; the number of active nodes or edges varies; or the population size fluctuates (in fact, the results in Supplementary Information allow for arbitrary patterns of replacement, see Supplementary Section 5). Although prosocial behaviors in different strategic domains may manifest in different ways, such as trust games or dictator games, the desire to pay costs to benefit others has a substantial degree of domain generality \cite{peysakhovich-nc-2014}. Our conclusions, based on donation games, are thus indicative of how dynamic networks may broadly impact prosocial behavior.

In the donation game, we have seen that changing social structures can sometimes promote cooperation. On the one hand, if the population simply transitions between homogeneous networks of low degree and high degree, then this dynamic tends to suppress cooperation \fig{spatial_temporal}\textbf{a,c}.  By contrast, if each network contains multiple communities of different local edge densities, then dynamic changes can facilitate cooperation and these effects can be dramatic (Figures \ref{fig:star-complete}-\ref{fig:other_structures}). Even if every such network individually disfavors cooperators, transitions between them can facilitate the evolution of cooperation -- a result that is reminiscent of Parrondo's paradox \cite{harmer:Nature:1999}. \fig{intuition} illustrates the mechanism for how this phenomenon arises, as transitions move individuals between regions of the network that are dense to those that are sparse. 
These dynamic structures also facilitate cooperation in the widely studied birth-death updating process (see Supplementary Figure 6) and for strong selection strength (see Supplementary Figure 7).

Changing social structures that each exhibit local communities are common in real-world settings. Groups and communities are more likely to form among people with close geographical locations and similar religion, culture, and affiliations \cite{girvan-2002-pnas,newman-2006-pnas}; but connection densities will be altered over time. Changes in connection densities in different communities may result from a phase difference, e.g. in online social networks across different time zones, in workers' interactions across a typical two-shift working schedule, and during periodic gatherings and migrations in religious communities. Spatio-temporal heterogeneity of interaction density within a given community also leads to time-varying connection densities, from sparse to dense and \emph{vice versa} \cite{2012-Holme-p97-125}. We find that each kind of burstiness has a clear effect on cooperation, either hindering it in the case of spatial burstiness or promoting it in the case of temporal burstiness. Broadly speaking, our work highlights the significance of integrating multiple communities into one system, since treating communities individually and independently may lead to erroneous conclusions about behavioral dynamics in the population as a whole \cite{Blonder-2012-mee}.

In line with many studies on evolutionary dynamics, we have assumed that mutations appear infrequently, which justifies the study of fixation probabilities to quantify evolutionary dynamics \cite{fudenberg:JET:2006}. Within static networks \cite{tarnita:JTB:2009} (and certain kinds of dynamic networks \cite{tarnita:PNAS:2009}), positive mutation rates lead to what is called a ``$\sigma$-rule'' for selection to favor cooperators in mean abundance. Such results involve inequalities based on (generally) at most three coefficients (e.g., $\sigma_{1}$, $\sigma_{2}$, and $\sigma_{3}$), all related to the population's structure, update rule, and mutation rate \cite{tarnita:PNAS:2011}. Calculating these coefficients is not straightforward, and only recently has a method been proposed for doing so on static networks \cite{mcavoy:PNAS:2022}. From a technical standpoint, the analysis of mutation-selection dynamics requires a different approach from what is typically used to study fixation probabilities. Furthermore, weak-mutation dynamics do not necessarily predict what happens when mutation rates are larger \cite{debarre:DGA:2019,mcavoy:PNAS:2022}. For these reasons, we have focused on mutation-free dynamics in this study. Of course, larger mutation rates may be relevant in dynamic populations, as they are in static populations, and future investigations into mutations on dynamic networks are warranted.

In addition to the mutation-free assumption, all of our results are based on exogenous network transitions, which means that individuals cannot selectively engineer their neighborhoods based on the traits of others. There are, of course, many interesting models involving endogenous transitions, in which cooperators can selectively form links with other cooperators and break links with defectors. In such models cooperation can flourish when structure transitions are rapid enough \cite{2006-Pacheco-p258103,2006-Santos-p1284-1291,2008-pacheco-jtb,2009-VanSegbroeck-p-,2010-Wu-p11187-11187,2011-Fehl-p546-551,2011-Rand-p19193-19198,Bravo2012,2012-Wang-p14363-14368,Bednarik2014,2014-Cardillo-p52825-52825,Harrell2018,2018-Akcay-p2692-2692}, for the simple reason that this endogenous dynamic establishes cooperative clusters. Such ``form follows function'' models are frequently aimed at answering the question: what kinds of networks arise from certain traits, and how do these networks serve the greater good? By contrast, our focus is not the coevolutionary dynamics of trait and structure, but on a different question altogether: what is the impact of exogenous structural changes on the evolution of behavior? This approach is more closely related to classical studies of network effects on cooperation: given a (dynamic) network, what behavioral traits evolve? Since exogenous structural changes do not provide any explicit advantage or disadvantage to cooperators relative to defectors, the resulting evolutionary dynamics of social traits are all the more intriguing.

We have aimed for generality in framing our mathematical results, but a natural limitation of our study is the scope of networks we have analyzed, compared to the vast space of possible population structures and transitions among them. For this reason, even static structures are still an active topic of current research in evolutionary game theory. We have therefore chosen to consider a limited number of representative examples of dynamic networks, which showcase the interesting effects they can have on the evolution of cooperation. Areas for future investigation include the effects of fluctuating resources on cooperation, varying the timescale of network changes relative to reproductive events (to include, for example, much more gradual changes), and other forms of exogenous changes that favor cooperation, and environments that involve both endogenous and exogenous transitions. In fact, although we use cooperation as an example, our analysis is framed quite generally to allow the study of other traits on dynamic structures. To the best of our knowledge, our analytical findings constitute the first general results for behavioral evolution on dynamic networks, and we hope that they will be valuable tools in future work.

\section{Methods}

\subsection{Analysis of weak selection}
Here, we outline a derivation of the critical benefit-to-cost ratio $\left(b/c\right)^{\ast}$ for selection to favor cooperation, based on an extension of a result of McAvoy \& Allen \cite{McAvoy2021}. Complete mathematical details of the extension may be found in Supplementary Sections 1-5. We first briefly recall the relevant result for (static) networks.

The general modeling framework of Allen \& McAvoy \cite{Allen2019} and McAvoy \& Allen \cite{McAvoy2021} assumes that the replacement process, which drives evolutionary dynamics, depends on the configuration of traits within the population, but it does not allow for a dynamic environmental state (e.g., network). In that context, a first-order expansion (in selection intensity) of a mutant trait's fixation probability was derived, which involves two quantities derived from neutral drift ($\delta =0$) and one quantity linking interactions to weak selection ($\delta\ll 1$). The first neutral quantity is reproductive value, $\pi$. The reproductive value of node $i$ is the probability that, under neutral drift, $i$ generates a lineage that takes over the population. In general, $\pi$ can be described as the unique solution to a linear system of $N$ equations. For the death-Birth (dB) update rule, we obtain $\pi_{i}\propto\sum_{j=1}^{N}w_{ij}$, where $w_{ij}$ is the weight of the edge between nodes $i$ and $j$ in the network.

The second quantity is related to the assortment of traits under neutral drift. Informally, for two nodes $i$ and $j$, we wish to understand the expected fraction of time prior to absorption that $i$ and $j$ share the same trait. This quantity, again, can be understood by looking at a linear system, this time of size $O\left(N^{2}\right)$. For dB updating, it suffices to calculate expected remeeting times of random walks on the network \cite{2017-Allen-p227-230}. Specifically, we seek quantities $\tau_{ij}$ such that $\tau_{ij}=0$ if $i=j$ and $\tau_{ij}=1+\frac{1}{2}\sum_{k=1}^{N}p_{ik}\tau_{kj}+\frac{1}{2}\sum_{k=1}^{N}p_{jk}\tau_{ik}$ when $i\neq j$, where $p_{ij}\coloneqq w_{ij}/\sum_{k=1}^{N}w_{ik}$ is the probability of moving from $i$ to $j$ in one step of a random walk on the network. Finally, the interaction comes into play via a game, in which type $A$ at node $i$ pays $C_{ij}$ to donate $B_{ij}$ to the individual at node $j$ (and type $B$ does nothing).

The fixation probability of a mutant placed uniformly-at-random in the network then satisfies
\begin{equation}\label{eq:main_fp_static}
\begin{split}
\frac{d}{d\delta}\Bigg\vert_{\delta =0} \rho_{C} &= \frac{1}{N}\sum_{i,j=1}^{N}\pi_{i}p_{ij}\sum_{\ell =1}^{N}\left(-\left(T-\tau_{jj}\right) C_{j\ell} +\left(T-\tau_{j\ell}\right) B_{\ell j}\right) \\
&\quad -\frac{1}{N}\sum_{i,j,k=1}^{N}\pi_{i}p_{ij}p_{ik}\sum_{\ell =1}^{N}\left(-\left(T-\tau_{jk}\right) C_{k\ell} +\left(T-\tau_{j\ell}\right) B_{\ell k}\right) ,
\end{split}
\end{equation}
where $T$ is the total time until all individuals are identical by descent. $T$ actually cancels out of \eq{main_fp_static}, but it is useful in the interpretation of this equation since then $T-\tau_{ij}$ is the mean time that nodes $i$ and $j$ are identical by descent.

We now outline how this result extends to dynamic networks. For $i,j\in\left\{1,\dots ,N\right\}$, let $w_{ij}^{\left[\beta\right]}$ be the weight of edge between nodes $i$ and $j$ in network $\beta\in\left\{1,\dots ,L\right\}$. We assume that the network is undirected, meaning $w_{ij}^{\left[\beta\right]}=w_{ji}^{\left[\beta\right]}$ for all $i,j\in\left\{1,\dots ,N\right\}$ and $\beta\in\left\{1,\dots ,L\right\}$. If $i$ and $j$ share an edge then they interact in game play and in strategic imitation. The class of models we are interested in here involve social goods \cite{mcavoy-2020-nhb} in which, on network $\beta$, an individual of type $A$ at $i$ pays a cost of $C_{ij}^{\left[\beta\right]}$ to donate $B_{ij}^{\left[\beta\right]}$ to the individual at $j$. In state $\left(\mathbf{x} ,\beta\right)$, the total payoff to the individual at $i$ is
\begin{equation}
u_{i}\left(\mathbf{x},\beta\right) = \sum_{j=1}^{N} \left( -x_{i} C_{ij}^{\left[\beta\right]} + x_{j} B_{ji}^{\left[\beta\right]} \right) . 
\end{equation} 
This net payoff is converted to fecundity via the formula $F_{k}\left(\mathbf{x},\beta\right) =e^{\delta u_{k}\left(\mathbf{x},\beta\right)}$. 
In the main text, under the assumption of weak selection, we use a linearized fecundity function $F_{k}\left(\mathbf{x},\beta\right)$. If the population structure is $\beta$, then a node in $\beta$ is first selected uniformly-at-random to die. Subsequently, all neighboring nodes in $\beta$ compete to produce an offspring to fill the vacancy at node $i$. The probability that $j$ replaces $i$ in state $\left(\mathbf{x} ,\beta\right)$ is given by \eq{eji_db}.

Let $p_{ij}^{\left[\beta\right]}\coloneqq w_{ij}^{\left[\beta\right]}/\sum_{k=1}^{N}w_{ik}^{\left[\beta\right]}$ be the probability of moving from $i$ to $j$ in one step of a random walk on network $\beta$. Under neutral drift, the probability $\pi_{i}^{\left[\beta\right]}$ that, starting in network $\beta$, $i$ generates a lineage that takes over the population (i.e.~the reproductive value of $i$ in $\beta$) satisfies
\begin{equation}
\pi_{i}^{\left[\beta\right]} = \frac{1}{N}\sum_{j=1}^{N} p_{ji}^{\left[\beta\right]} \sum_{\gamma =1}^{L}q_{\beta\gamma}\pi_{j}^{\left[\gamma\right]}
+\left(1-\frac{1}{N}\sum_{j=1}^{N}p_{ij}^{\left[\beta\right]}\right)\sum_{\gamma =1}^{L} q_{\beta\gamma}\pi_{i}^{\left[\gamma\right]} , \label{eq:reproductive_value_db}
\end{equation}
subject to the constraint $\sum_{i=1}^{N}\pi_{i}^{\left[\beta\right]}=1$. $\pi$ is thus determined by a linear system of size $O\left(LN\right)$.
The reproductive value of individual $i$ in network $\beta$ satisfies a system involving the sum of two components. The first term represents the reproductive value of the offspring, if $i$ spreads its offspring to another node in the network $\gamma$. The second term represents the reproductive value of individual $i$ in the subsequent network $\gamma$, if it is not selected to die.

For the initial state, we choose a network from the stationary distribution of stochastic network transitions; and a mutant is chosen uniformly-at-random within that network. We consider two types of initial mutants:  a mutant $C$ arising in an otherwise all-$D$ state (denoted $\mu_{C}$) and a mutant $D$ arising in an otherwise all-$C$ state (denoted $\mu_{D}$). Associated with each $\mu\in\left\{\mu_{C},\mu_{D}\right\}$ is a quantity $\eta_{I}^{\left[\beta\right]}\left(\mu\right)$ related to the co-occurrence of a trait in $\beta$ among the nodes in $I\subseteq\left\{1,\dots ,N\right\}$, which is defined formally in Supplementary Section 5. We can think of $\eta_{I}^{\left[\beta\right]}$ as representing the probability that, given the initial state, all individuals in $I$ have type $C$. For our purposes, we need an expression for $\eta_{I}^{\left[\beta\right]}\left(\mu\right)$ only when $I$ contains one or two nodes, because interactions occur among pairs of neighbors. (More details on the required size of $I$ may be found in Supplementary Section 5, and an extensive discussion of how the maximal size of $I$ is determined by the \emph{degree} of a game may be found in McAvoy \& Allen \cite{McAvoy2021}.) For $\left| I\right| =2$, we find that $\eta_{I}^{\left[\beta\right]}\left(\mu_{C}\right) =\eta_{I}^{\left[\beta\right]}\left(\mu_{D}\right)$ and
\begin{equation}
\eta_{ij}^{\left[\beta\right]} = 
\begin{cases}
\displaystyle 0 & \displaystyle i=j , \\
& \\
\displaystyle \frac{1}{N}\upsilon\left(\beta\right) +\sum_{\gamma =1}^{L}q_{\gamma\beta}\left(\frac{1}{N}\sum_{k=1}^{N}p_{ik}^{\left[\gamma\right]} \eta_{kj}^{\left[\gamma\right]} +\frac{1}{N}\sum_{k=1}^{N}p_{jk}^{\left[\gamma\right]} \eta_{ik}^{\left[\gamma\right]} +\left(1-\frac{2}{N}\right)\eta_{ij}^{\left[\gamma\right]}\right) & \displaystyle i\neq j .
\end{cases} \label{eq:eta_db}
\end{equation}
We refer the reader to Equation 32 in Supplementary Information for further details.

It turns out that a scaled version of $\eta$, namely $\tau_{ij}^{\left[\beta\right]}\coloneqq\eta_{ij}^{\left[\beta\right]}/\upsilon\left(\beta\right)$, allows for a more intuitive interpretation of the selection condition. Consider the time-reversed structure transition chain defined by
\begin{equation}
\widetilde{q}_{\beta\gamma} \coloneqq \frac{\upsilon\left(\gamma\right)}{\upsilon\left(\beta\right)} q_{\gamma\beta} .
\end{equation}
Using this time-reversed chain in conjunction with \eq{eta_db}, we see that
\begin{equation}
\tau_{ij}^{\left[\beta\right]} = 
\begin{cases}
\displaystyle 0 & \displaystyle i=j , \\
& \\
\displaystyle \frac{1}{N} +\sum_{\gamma =1}^{L}\widetilde{q}_{\beta\gamma}\left(\frac{1}{N}\sum_{k=1}^{N}p_{ik}^{\left[\gamma\right]} \tau_{kj}^{\left[\gamma\right]} +\frac{1}{N}\sum_{k=1}^{N}p_{jk}^{\left[\gamma\right]} \tau_{ik}^{\left[\gamma\right]} +\left(1-\frac{2}{N}\right)\tau_{ij}^{\left[\gamma\right]}\right) & \displaystyle i\neq j .
\end{cases} \label{eq:tau_db}
\end{equation}
In the ancestral process, looking backward in time under neutral drift, $N\tau_{ij}^{\left[\beta\right]}$ represents the expected number of update steps until $i$ and $j$ coalesce. Equivalently, since one of $N$ individuals is updated in each time step, $\tau_{ij}^{\left[\beta\right]}$ can be seen as the mean number of generations needed for $i$ and $j$ to coalesce. This interpretation arises from a one-step analysis of the time-reversed chain and the fact that, in each time increment, one update step (a portion $1/N$ of a generation) has transpired, which accounts for the additive factor of $1/N$ in \eq{tau_db}. The probability of transitioning from network $\beta$ to network $\gamma$ backward in time (i.e., in the \emph{reversed} chain) is $\widetilde{q}_{\beta\gamma}$. With probability $1/N$, $i$ is chosen for replacement, and with probability $p_{ik}$, the offspring of $k$ replaces $i$. After this step, the mean number of generations for the resulting random walks to coalesce is $\tau_{kj}^{\left[\gamma\right]}$. Likewise, with probability $1/N$, $j$ is chosen for replacement, and with probability $p_{jk}$, the offspring of $k$ replaces $j$. Then, the mean number of generations for coalescence is $\tau_{ik}^{\left[\gamma\right]}$. Finally, with probability $1-2/N$, neither $i$ nor $j$ is chosen for replacement. Since the network can still transition to $\gamma$ even when neither $i$ nor $j$ is replaced, we are left with the mean coalescence time $\tau_{ij}^{\left[\gamma\right]}$. We refer the reader to Allen et al.~\cite{2017-Allen-p227-230} and Allen \& McAvoy \cite{allen:preprint:2022} for more detailed discussions of random walks on networks and their relationships to the coalescent.

If, conditioned on the population being in state $\beta$, $T^{\left[\beta\right]}$ is the mean time to reach the most recent common ancestor going backward in time, then the mean time that $i$ and $j$ spend identical by descent is $T^{\left[\beta\right]}-\tau_{ij}^{\left[\beta\right]}$. Finding $\tau$ for all structures and pairs of individuals involves solving a linear system of size $O\left(LN^{2}\right)$. We note that, although $T$ aids in the interpretation of $\tau$ as determining identity by descent, it does not need to be calculated directly in order to understand the first-order effects of selection on fixation probability.

We now have all of the neutral quantities we need to state the selection condition. The final piece is the connection between the payoffs and the replacement probabilities under weak selection. A straightforward calculation gives the probability that individual $i$ copies $j$ type, i.e.  $e_{ji}\left(\mathbf{x},\beta\right) =\frac{1}{N}p_{ij}^{\left[\beta\right]}+\delta\sum_{k=1}^{N}c_{k}^{ji}\left(\beta\right) x_{k}+O\left(\delta^{2}\right)$, where
\begin{equation}\label{eq:cji_beta_db_game}
c_{k}^{ji}\left(\beta\right) = 
\begin{cases}
\displaystyle \frac{1}{N}p_{ij}^{\left[\beta\right]}\left( -\sum_{\ell =1}^{N}C_{j\ell}^{\left[\beta\right]}+B_{jj}^{\left[\beta\right]}+p_{ij}^{\left[\beta\right]}\sum_{\ell =1}^{N}C_{j\ell}^{\left[\beta\right]}
-\sum_{\ell =1}^{N}p_{i\ell}^{\left[\beta\right]}B_{j\ell}^{\left[\beta\right]}\right) & \displaystyle k= j , \\
& \\
\displaystyle \frac{1}{N}p_{ij}^{\left[\beta\right]}\left( B_{kj}^{\left[\beta\right]}+p_{ik}^{\left[\beta\right]}\sum_{\ell =1}^{N}C_{k\ell}^{\left[\beta\right]}
-\sum_{\ell =1}^{N}p_{i\ell}^{\left[\beta\right]}B_{k\ell}^{\left[\beta\right]}\right) & \displaystyle k\neq j .
\end{cases}
\end{equation}

Putting everything together using Equation 31 in Supplementary Information, we see that
\begin{equation}
\begin{split}
\frac{d}{d\delta}\Bigg\vert_{\delta =0} \rho_{C}&= \frac{1}{N}\sum_{i,j=1}^{N} \sum_{\beta =1}^{L} \upsilon\left(\beta\right) \left(\sum_{\gamma =1}^{L} q_{\beta\gamma} \pi_{i}^{\left[\gamma\right]}\right) p_{ij}^{\left[\beta\right]} \sum_{\ell =1}^{N} \left(\substack{-\left(T^{\left[\beta\right]}-\tau_{jj}^{\left[\beta\right]}\right) C_{j\ell}^{\left[\beta\right]} \\ + \left(T^{\left[\beta\right]}-\tau_{j\ell}^{\left[\beta\right]}\right) B_{\ell j}^{\left[\beta\right]}}\right)  \\
&\quad -\frac{1}{N}\sum_{i,j,k=1}^{N} \sum_{\beta =1}^{L} \upsilon\left(\beta\right) \left(\sum_{\gamma =1}^{L} q_{\beta\gamma} \pi_{i}^{\left[\gamma\right]}\right) p_{ij}^{\left[\beta\right]} p_{ik}^{\left[\beta\right]} \sum_{\ell =1}^{N} \left( \substack{-\left(T^{\left[\beta\right]}-\tau_{jk}^{\left[\beta\right]}\right) C_{k\ell}^{\left[\beta\right]}
\\ +\left(T^{\left[\beta\right]}-\tau_{j\ell}^{\left[\beta\right]}\right) B_{\ell k}^{\left[\beta\right]}}\right) . \label{eq:sg_derivative}
\end{split}
\end{equation}
An analogous calculation for $D$ gives $\frac{d}{d\delta}\Big\vert_{\delta =0}\rho_{D}=-\frac{d}{d\delta}\Big\vert_{\delta =0}\rho_{C}$, which means $\frac{d}{d\delta}\Big\vert_{\delta =0}\rho_{C}>\frac{d}{d\delta}\Big\vert_{\delta =0}\rho_{D}$ is equivalent to $\frac{d}{d\delta}\Big\vert_{\delta =0}\rho_{C}>0$.
Under weak selection ($\delta\ll 1$) we have
$\rho_C=\rho_C\Big\vert_{\delta=0}+\delta \frac{d}{d\delta}\Big\vert_{\delta=0}\rho_C+O(\delta^2)$ and 
$\rho_D=\rho_D\Big\vert_{\delta=0}+\delta \frac{d}{d\delta}\Big\vert_{\delta=0}\rho_D+O(\delta^2)$.
The quantities $\rho_C\Big\vert_{\delta=0}$ and $\rho_D\Big\vert_{\delta=0}$ are respectively the fixation probability of cooperators and defectors under neutral drift, and they are identical, i.e. $\rho_C\Big\vert_{\delta=0}=\rho_D\Big\vert_{\delta=0}=1/N$.
In summary, the condition for selection to favor cooperation, $\rho_C>\rho_D$, is therefore equivalent to $\frac{d}{d\delta}\Big\vert_{\delta=0}\rho_C>\frac{d}{d\delta}\Big\vert_{\delta=0}\rho_D$ and accordingly $\frac{d}{d\delta}\Big\vert_{\delta =0}\rho_{C}>0$.

In the donation game, we have $B_{ij}^{\left[\beta\right]}=w_{ij}^{\left[\beta\right]}b$ and $C_{ij}^{\left[\beta\right]}=w_{ij}^{\left[\beta\right]}c$. 
Collecting all terms involving $T$ in \eq{sg_derivative} gives $0$, and then we have 
\begin{equation} \label{eq:rhoAB_stochastic2}
\rho_C>\rho_D \iff b\mu_{2}-c\nu_{2}>b\mu_{0}-c\nu_{0} , 
\end{equation}
where 
\begin{equation} \label{eq:uv_stochastic1}
\begin{split}
\mu_{0} &= \frac{1}{N}\sum_{i,j=1}^{N} \sum_{\beta =1}^{L} \upsilon\left(\beta\right) \left(\sum_{\gamma =1}^{L} q_{\beta\gamma} \pi_{i}^{\left[\gamma\right]}\right) p_{ij}^{\left[\beta\right]} \sum_{\ell =1}^{N} w_{\ell j}^{\left[\beta\right]}\tau_{j\ell}^{\left[\beta\right]} ; \\
\nu_{0} &= \frac{1}{N}\sum_{i,j=1}^{N} \sum_{\beta =1}^{L} \upsilon\left(\beta\right) \left(\sum_{\gamma =1}^{L} q_{\beta\gamma} \pi_{i}^{\left[\gamma\right]}\right) p_{ij}^{\left[\beta\right]} \sum_{\ell =1}^{N} w_{j\ell}^{\left[\beta\right]}\tau_{jj}^{\left[\beta\right]} ; \\
\mu_{2} &= \frac{1}{N}\sum_{i,j,k=1}^{N} \sum_{\beta =1}^{L} \upsilon\left(\beta\right) \left(\sum_{\gamma =1}^{L} q_{\beta\gamma} \pi_{i}^{\left[\gamma\right]}\right) p_{ij}^{\left[\beta\right]} p_{ik}^{\left[\beta\right]} \sum_{\ell =1}^{N} w_{\ell k}^{\left[\beta\right]}\tau_{j\ell}^{\left[\beta\right]} ; \\
\nu_{2} &= \frac{1}{N}\sum_{i,j,k=1}^{N} \sum_{\beta =1}^{L} \upsilon\left(\beta\right) \left(\sum_{\gamma =1}^{L} q_{\beta\gamma} \pi_{i}^{\left[\gamma\right]}\right) p_{ij}^{\left[\beta\right]} p_{ik}^{\left[\beta\right]} \sum_{\ell =1}^{N} w_{k\ell}^{\left[\beta\right]}\tau_{jk}^{\left[\beta\right]} .
\end{split}
\end{equation}
The critical benefit-to-cost ratio is therefore $\left(b/c\right)^{\ast}=\left(\nu_{2}-\nu_{0}\right) /\left(\mu_{2}-\mu_{0}\right)$.

Note that, for simplicity, we have assumed that any node can be selected for death in a given network. In reality, this assumption might not hold because each individual network need not be connected, which can lead to isolated nodes. If an isolated node is chosen for death, then the individual at this node cannot be immediately replaced by the offspring of a neighbor. All of our calculations can be modified to allow for only non-isolated nodes to be chosen for death, although in practice we do not need to do so in any of our examples.

\subsection{Specific examples}
We study the transition between two networks, with transition probabilities given by
\begin{equation} \label{eq:sparse_dense3}
\begin{split}
q_{\beta\gamma} &= 
\begin{cases}
\displaystyle 1-p & \displaystyle \beta =1,\gamma =1 ; \\
\displaystyle p & \displaystyle \beta =1,\gamma =2 ; \\
\displaystyle q & \displaystyle \beta =2,\gamma =1 ; \\
\displaystyle 1-q & \displaystyle \beta =2,\gamma =2 .
\end{cases} 
\end{split}
\end{equation}
The expected durations in networks $1$ and $2$ are $q/\left(p+q\right)$ and $p/\left(p+q\right)$, respectively.

We study evolution on dynamic two-community networks. The two-community network is made up of a star community and a complete community, with the hubs connected (see \fig{star-complete}\textbf{a}). Let $n$ and $m$ denote the numbers of nodes in the star and complete communities, respectively, so that $n+m=N$. We denote by $1,\dots ,n$ the nodes in the star community and by $n+1,\dots n+m$ the nodes in the complete community, where $n$ is the hub of the star and $n+m$ is the node of the complete community connected to the hub of the star. The other network is obtained by swapping the star and complete communities. The adjacency matrix for the first network satisfies $w_{ij}^{\left[1\right]}=1$ only if \emph{(i)} $i=n$ and $j<n$, or $i<n$ and $j=n$; \emph{(ii)} $i=n$ and $j=n+m$, or $i=n+m$ and $j=n$; or \emph{(iii)} $i,j\geqslant n+1$ and $i\neq j$. The adjacency for the second network satisfies $w_{ij}^{\left[2\right]}=1$ only if \emph{(i)} $i=n+1$ and $j>n+1$, or $i>n+1$ and $j=n+1$; \emph{(ii)} $i=n$ and $j=n+m$, or $i=n+m$ and $j=n$; or \emph{(iii)} $i,j\leqslant n$ and $i\neq j$.

Using the results of the previous section, we can directly calculate $\pi$ and $\tau$ and calculate the critical benefit-to-cost ratio. Here, we provide explicit mathematical results for representative cases. Assuming $p=q=1/\left(tN\right)$ and letting $a\coloneqq n/\left(n+m\right)$, we find that 
\begin{equation} \label{eq:sparse_dense_dynamic}
\begin{split}
\left(\frac{b}{c}\right)^{\ast} &= 
\begin{cases}
\frac{\left(2a^{2} - 2a + 1\right) t^{3} + \left(8a^{2} - 8a + 7\right) t^{2} + \left(8a^{2} - 8a + 15
\right) t + 2a^{2} - 2a + 10}{2a\left(1-a\right)\left(t^{2} + 4t + 3\right)} & N\rightarrow \infty , \\
& \\
\frac{t^{3} + 10t^{2} + 26t + 19}{t^{2} + 4t + 3} & N\rightarrow \infty, \, a=\frac{1}{2}, \\
& \\
\frac{20a^{2} - 20a + 33}{16a\left(1-a\right)} & N\rightarrow \infty\, , t=1, \\
& \\
\frac{\tilde{t} \left(\tilde{t} + 1\right)}{ - 2\tilde{t}^{2} + 2\tilde{t} + 1}N & N\rightarrow \infty, \, t/N=\tilde{t}, \, a=\frac{1}{2}, \\
& \\
\frac{\left(\substack{210N^{16} - 520N^{15} - 1034N^{14} + 1770N^{13} \\ 
+ 14028N^{12} - 93440N^{11} + 300848N^{10} - 330944N^{9} \\
- 663040N^{8} + 2230528N^{7} - 1096448N^{6} - 4570112N^{5}  \\ 
+ 10000384N^{4} - 9265152N^{3} + 4259840N^{2} - 786432N}\right)}
{\left(\substack{30N^{16} - 79N^{15} + 225N^{14} + 1756N^{13}\\ - 15088N^{12} - 13128N^{11} + 247296N^{10} - 365152N^{9} \\
- 849344N^{8} + 2987392N^{7} - 1801984N^{6} - 5024768N^{5} \\
+ 11302912N^{4} - 9949184N^{3} + 3940352N^{2} - 327680N - 131072}\right)} & a=\frac{1}{2},\, t=1 .
\end{cases} 
\end{split}
\end{equation}
In particular, when $N\rightarrow\infty$ and $a=1/2$, $\left(b/c\right)^{\ast}$ is a monotonically increasing function of $t$.

We can compare this critical ratio to that of just a single network, which is the same for either network $1$ or network $2$ and satisfies
\begin{equation} \label{eq:sparse_dense_single}
\begin{split}
\left(\frac{b}{c}\right)^{\ast} &= 
\begin{cases}
-\left(1-a\right) N & N\rightarrow \infty , \\
& \\
\frac{- 3N^{9} - 40N^{8} - 204N^{7} - 848N^{6} - 2464N^{5} + 1920N^{4} + 15872N^{3} - 40960N^{2} + 24576N}{6N^{8} - 64N^{7} - 520N^{6} - 1232N^{5} + 3872N^{4} + 6272N^{3} - 24320N^{2} + 22528N + 4096}
& a=\frac{1}{2} .
\end{cases} 
\end{split}
\end{equation}

\subsection*{Acknowledgements}
This work is supported by the Simons Foundation (Math+X Grant to the University of Pennsylvania). J.B.P. acknowledges generous support by the John Templeton Foundation (grant no. 62281), and the David \& Lucille Packard Foundation.

\clearpage

\makeatletter
\@fpsep\textheight
\makeatother

\setcounter{figure}{0}
\renewcommand{\figurename}{Supplementary Figure}

\begin{figure*}[!h]
\centering
\includegraphics[width=0.5\textwidth]{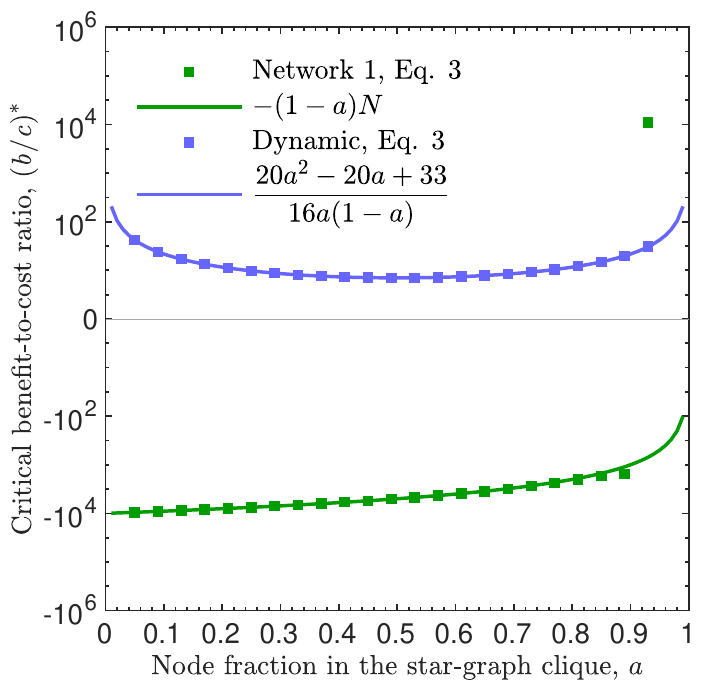}
\caption{ \textbf{The cooperation-promoting effects of structure transitions as the sizes of the two communities vary.} The dynamic network is illustrated in Figure 2\textbf{a}, with a fraction $a$ (resp. $1-a$) of nodes in the top (resp. bottom) community. The critical benefit-to-cost ratio, $\left(b/c\right)^{\ast}$, is shown as a function of $a$. The dots are the results of numerical calculations with $N=10{,}000$ and the lines are analytical approximations for sufficiently large $N$. The rescaled duration is $t=1$.\label{fig:ED_Fig1}}
\end{figure*}

\newpage

\newpage
\begin{figure*}
\centering
\includegraphics[width=1.0\textwidth]{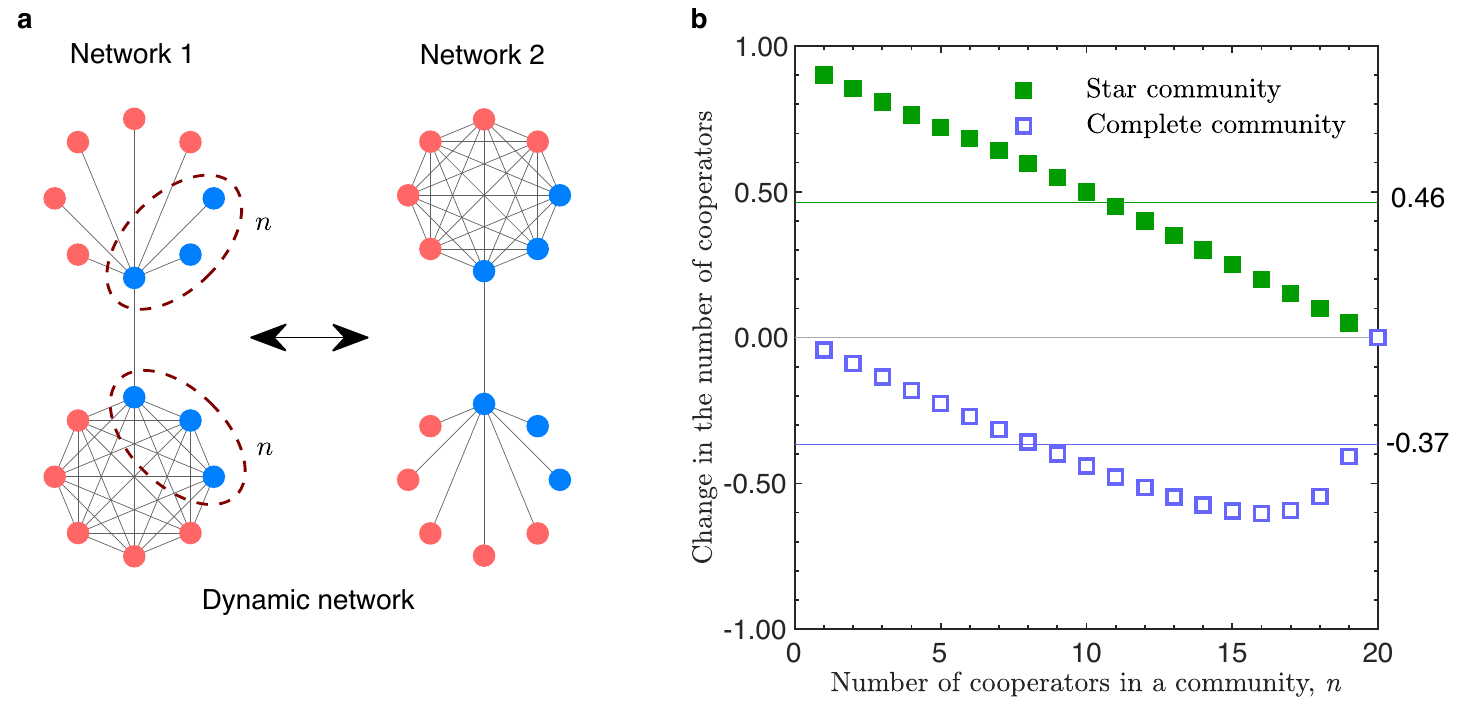}
\caption {{\textbf{A qualitative analysis of cooperator expansion in the star community 
and cooperator reduction in the complete community.}}
\textbf{a}, We consider a configuration where the hub node and $n-1$ other nodes are cooperators in each community, and the rest are defectors.
The fecundity of the hub node in the star community is given by $F_{\text{star},\text{hub}}=\exp\left[\delta\left(nb-Nc/2\right)\right]$. And the fecundity of a cooperator the other nodes is given by $F_{\text{star},C}=\exp\left[\delta\left(b-c\right)\right]$; whereas the fecundity of a defector node is given by $F_{\text{star},D}=\exp\left[\delta b\right]$). We can also derive expressions for fecundities in the complete community: $F_{\text{comp},\text{hub}}=\exp\left[\delta\left(nb-Nc/2\right)\right]$, $F_{\text{comp},C}=\exp\left[\delta\left((n-1)b-(N/2-1)c\right)\right]$, $F_{\text{comp},D}=\exp\left[\delta nb\right]$. Finally, we can calculate the expected change in the number of cooperators in each community, denoted by $\Delta_{\text{star}}$ in the star community and $\Delta_{\text{comp}}$ in the complete community.
\textbf{b}, The expected change in the number of cooperators in the star community and in the complete community, for across the full range sub-graph sizes, $n$.
Dots are obtained by $\Delta_{\text{star}}$ and $\Delta_{\text{comp}}$.
Horizontal lines mark the average change across all possibilities of $n$.
The increase in the number of cooperators in the star community exceeds the decrease in the number of cooperators in the complete community.
Parameters: $N=40$, $\delta=0.1$, $b=10$, $c=1$.
\label{fig:Assumption}}
\end{figure*}

\newpage

\begin{figure*}[!h]
\centering
\includegraphics[width=1\textwidth]{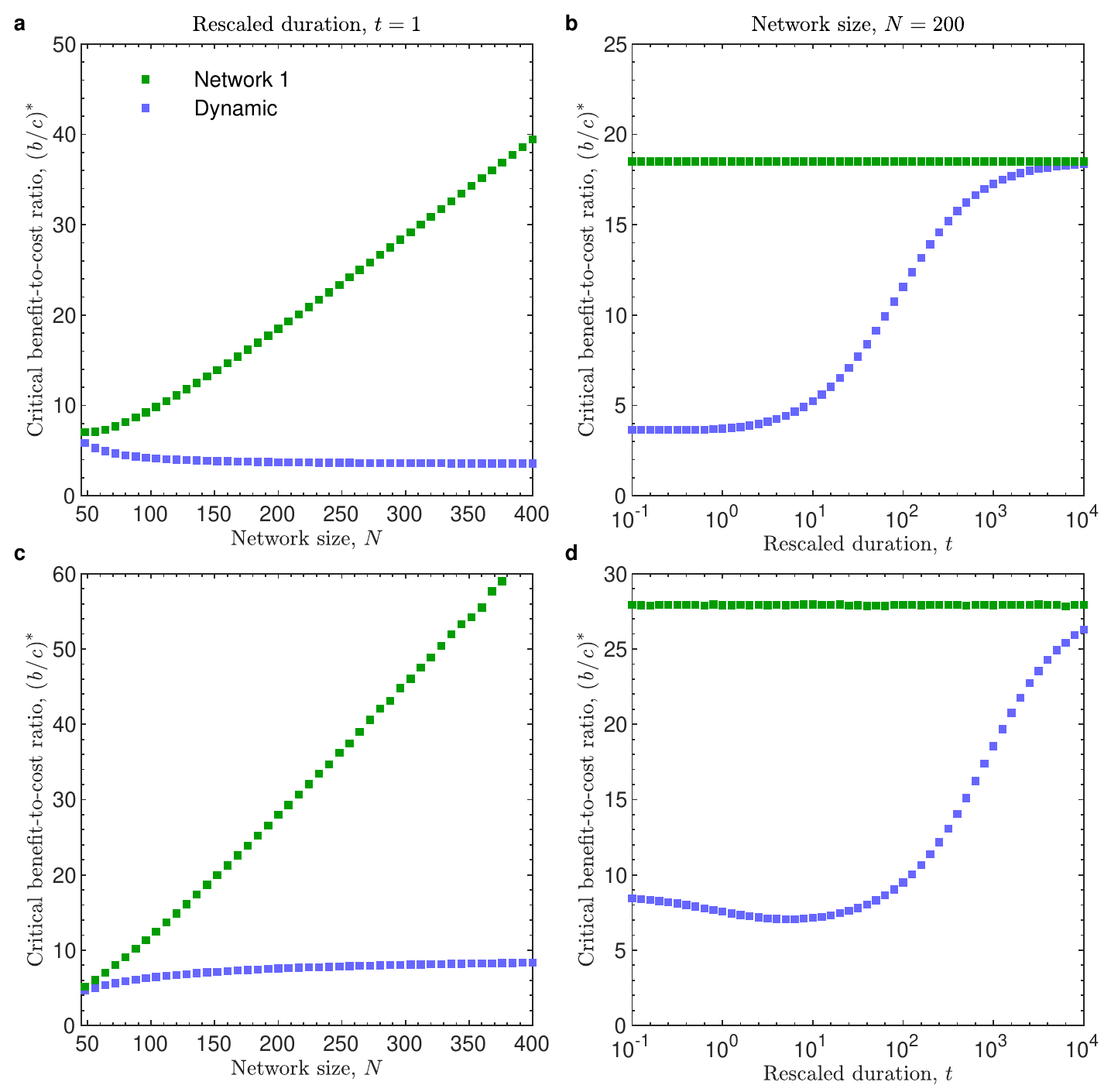}
\caption{\textbf{Cooperation-promoting effects of dynamic multi-community networks.} We consider networks made up of eight communities connected via hub nodes (see Figure 5\textbf{a}; panels \textbf{a} and \textbf{b} here) and via leaf nodes (see Figure 5\textbf{b}; panels \textbf{c} and \textbf{d} here). \textbf{a,c}, The critical ratio $\left(b/c\right)^{\ast}$ as a function of population size $N$, for the rescaled duration $t=1$. \textbf{b,d}, The critical ratio $\left(b/c\right)^{\ast}$ as a function of the rescaled duration $t$, for $N=200$.\label{fig:multi_community_size}}
\end{figure*}

\newpage

\begin{figure*}[!h]
\centering
\includegraphics[width=1\textwidth]{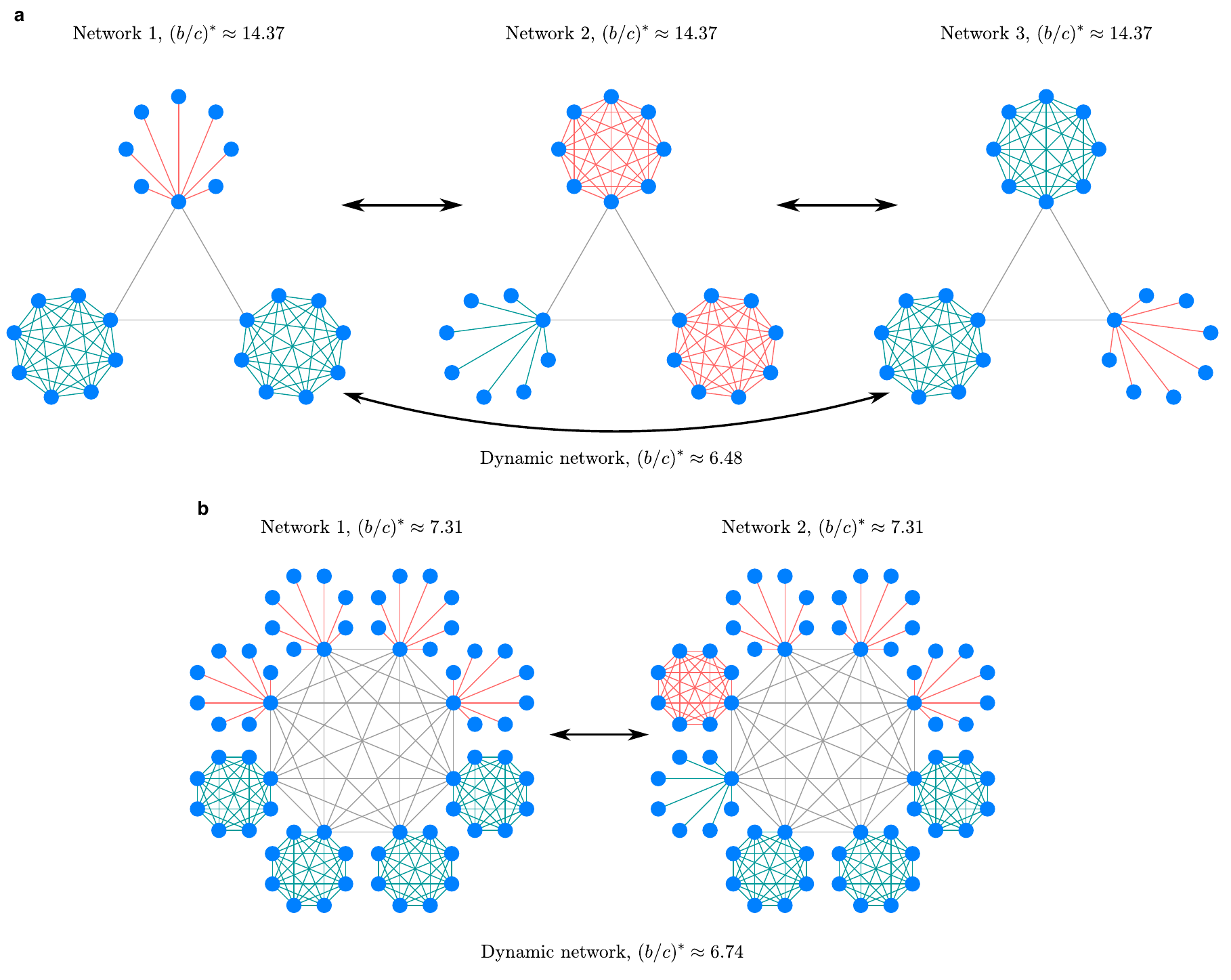}
\caption{\textbf{Cooperation-promoting effects of structure transitions among more than two networks, and when networks differ in a small fraction of connections.} \textbf{a}, Structure transitions among three networks. Every network transitions to another network with probability $1/\left(2tN\right)$ and remains unchanged otherwise. \textbf{b}, Structure transitions between multi-community networks in which the two networks differ in only two communities. We take $N=150$ in \textbf{a} and $N=64$ in \textbf{b}, and the rescaled duration is $t=1$.\label{fig:similar_networks}}
\end{figure*}

\newpage

\begin{figure*}[!h]
\centering
\includegraphics[width=1\textwidth]{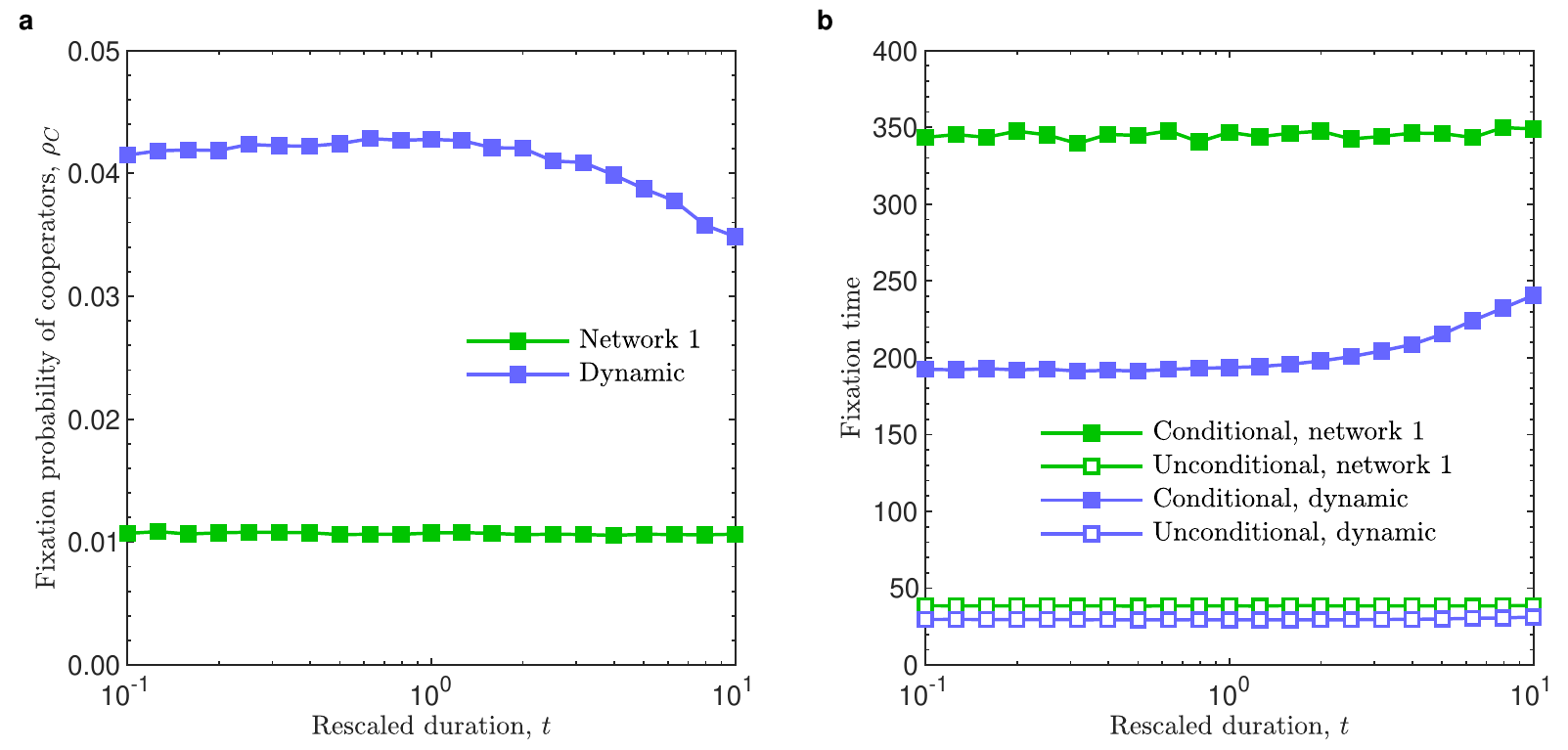}
\caption{\textbf{Dynamic networks promote and accelerate the fixation of cooperators.} We consider the network with a star community and a complete-graph community with $N=16$ and $a=0.5$ (see Figure 2\textbf{a}). \textbf{a}, Fixation probability of cooperators as a function of the rescaled duration, $t$, in network $1$ and in the dynamic network. The dynamic network leads to the larger fixation probability of cooperators than in network $1$. \textbf{b}, Conditional and unconditional fixation times as functions of the rescaled duration, $t$. Both the conditional and unconditional times in the dynamic networks are smaller than in network $1$.We take selection intensity $\delta =0.1$. \label{fig:time}}
\end{figure*}

\newpage

\begin{figure*}
\centering
\includegraphics[width=0.6\textwidth]{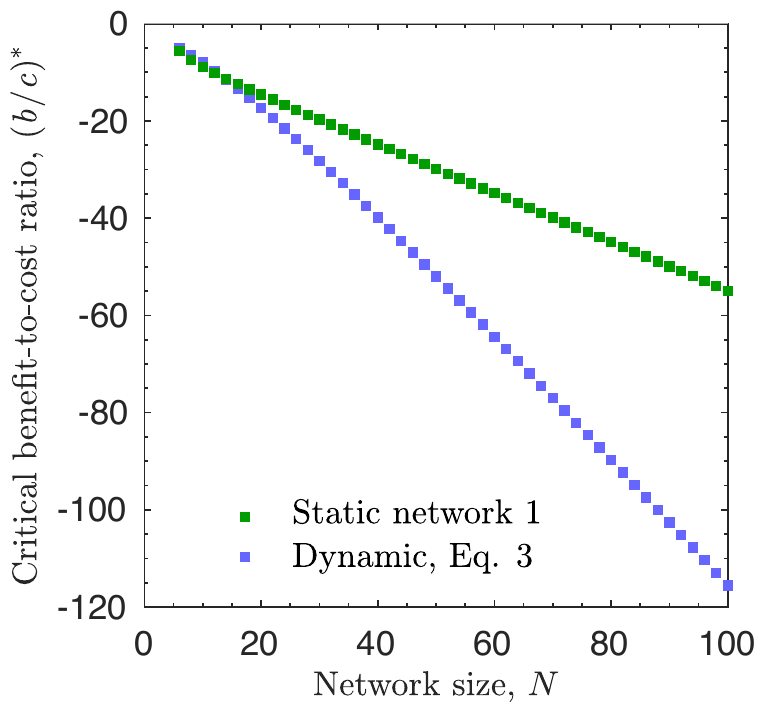}
\caption {{\bf{Dynamic structures can inhibit spite for a broad range of population sizes, under birth-death updating.} }
We consider transitions between the two networks shown in Figure 2{\bf{a}} in the main text, each composed of a sparse community and a dense community. 
The figure plots the critical benefit-to-cost ratio for cooperation as a function of population size, $N$, for $a=0.5$ and $t=1$.
Under birth-death updating, spite rather than cooperation can evolve, as indicated by a negative critical ratio $(b/c)^*$.
These dynamic networks, compared with static networks,  reduce the magnitude of the critical ratio, thus making spiteful behavior harder to evolve.
\label{fig:birth_death}}
\end{figure*}

\newpage

\begin{figure*}
\centering
\includegraphics[width=0.6\textwidth]{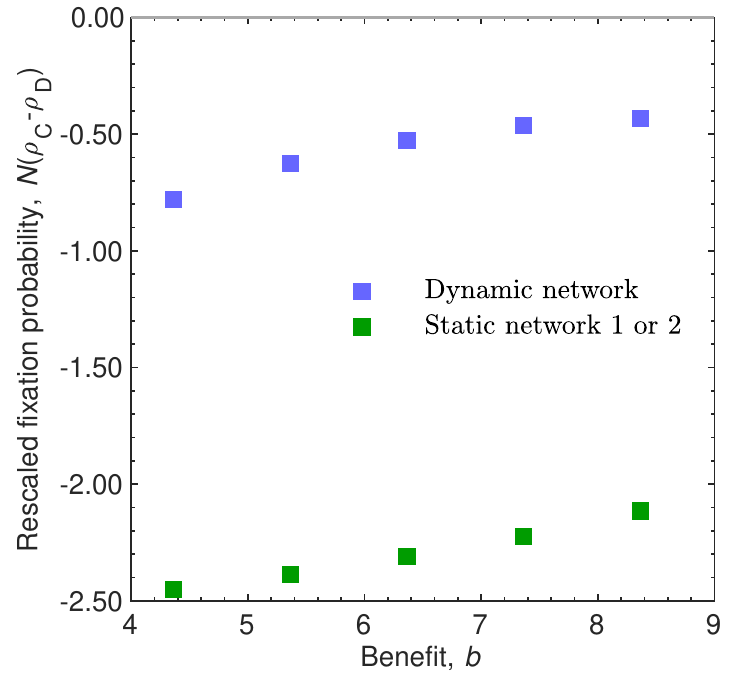}
\caption {{\textbf{Dynamic structures can facilitate cooperation under strong selection.} 
We consider transitions between the two networks shown in Figure 2{\bf{a}} in the main text, each composed of a sparse community and a dense community. 
The figure plots the fixation probability of cooperation versus defection, $\rho_{C}-\rho_{D}$, as a function of the benefit $b$ in the donation game, for $N=12$, $a=0.5$, and $t=1$.
These dynamic networks are more beneficial than static networks, for the evolution of cooperation. Selection strength: $\delta=0.1$.
}
\label{fig:selection_strength}}
\end{figure*}

\clearpage

\begin{center}
\Large\textbf{Supplementary Information}
\end{center}

\setcounter{equation}{0}
\setcounter{figure}{0}
\setcounter{section}{0}
\setcounter{table}{0}
\renewcommand{\thesection}{SI.\arabic{section}}
\renewcommand{\theequation}{SI.\arabic{equation}}
\renewcommand{\thefigure}{SI.\arabic{figure}}

\section{Modeling evolution on dynamic networks}
\label{sec:assumptions}
We consider a population of $N$ individuals (labeled $\mathcal{N}=\left\{1,2,\dots ,N\right\}$), residing at any point in time on one of $L$ structures (labeled $\mathcal{L}=\left\{1,2,\dots ,L\right\}$). Implicitly, this means that each of these $L$ structures is a network on $N$ nodes, although each network need not be connected, and some nodes can be isolated. Each individual has type $A$ or $B$, and the state of population is tracked by a pair $\left(\vx ,\beta\right)\in\left\{0,1\right\}^{\mathcal{N}}\times\mathcal{L}$, where $x_{i}=1$ means $i$ has type $A$ and $x_{i}=0$ means $i$ has type $B$.

At each time step, a set of individuals to be replaced, $R\subseteq\mathcal{N}$, is chosen, together with an offspring-to-parent map, $\alpha :R\rightarrow\mathcal{N}$. Let $p_{\left(R,\alpha\right)}\left(\vx ,\beta\right)$ denote the probability of replacement event $\left(R,\alpha\right)$ in state $\left(\vx ,\beta\right)$. Once $\left(R,\alpha\right)$ is chosen, the type configuration, $\vx$, is updated to $\vy$, where $y_{i}=x_{\alpha\left(i\right)}$ if $i\in R$ and $y_{i}=x_{i}$ if $i\not\in R$. This update can be specified more succinctly using an extended mapping $\widetilde{\alpha}:\mathcal{N}\rightarrow\mathcal{N}$ defined by $\widetilde{\alpha}\left(j\right) =\alpha\left(j\right)$ if $j\in R$ and $\widetilde{\alpha}\left(j\right) =j$ if $j\not\in R$, which leads to the updated state $\vx_{\widetilde{\alpha}}$, where $\left(\vx_{\widetilde{\alpha}}\right)_{i}=x_{\widetilde{\alpha}\left(i\right)}$ for $i\in\mathcal{N}$. The network, $\beta$, is updated via a transition matrix, $Q=\left(q_{\beta\gamma}\right)_{\beta,\gamma\in\mathcal{L}}$, where $q_{\beta\gamma}$ is the probability of transitioning from network $\beta$ to network $\gamma$. An important feature of the model is that network transitions are independent of $\vx$; thus, the population structure is exogenous and not influenced by traits. We assume that $Q$ is irreducible, which guarantees that it has a unique stationary distribution, $\upsilon$.

We assume that for each replacement event, $\left(R,\alpha\right)$, type configuration, $\vx$, and network, $\beta$, the probability $p_{\left(R,\alpha\right)}\left(\vx ,\beta\right)$ is a smooth function of a selection intensity parameter, $\delta\geqslant 0$, in a small neighborhood of $\delta =0$. Moreover, when $\delta =0$ (``neutral drift''), we assume that $p_{\left(R,\alpha\right)}\left(\vx ,\beta\right)$ is independent of $\vx$ (but it can depend on $\beta$). We denote by $p_{\left(R,\alpha\right)}^{\circ}\left(\beta\right)$ the probability of choosing $\left(R,\alpha\right)$ under neutral drift. The chain defined by $Q$ does not depend on the selection intensity.

We also make the following assumption, which ensures that for every starting configuration and network, there exists at least one individual whose lineage can take over the population:
\begin{fixation}
For all network structures $\beta_{0}\in\mathcal{L}$, there exists a location $i\in\mathcal{N}$, an integer $m \geqslant 1$, and sequences of replacement events $\left\{\left(R_{k},\alpha_{k}\right)\right\}_{k=1}^{m}$ and networks $\left\{\beta_{k}\right\}_{k=1}^{m-1}$ for which
\begin{enumerate}

\item[\emph{(i)}] $p_{\left(R_{k},\alpha_{k}\right)}\left(\vx ,\beta_{k-1}\right) >0$ for every $k\in\left\{1,\dots ,m\right\}$ and $\vx\in\left\{0,1\right\}^{\mathcal{N}}$;

\item[\emph{(ii)}] $q_{\beta_{k-1}\beta_{k}}>0$ for every $k\in\left\{1,\dots ,m-1\right\}$;

\item[\emph{(iii)}] $i\in R_{k}$ for some $k\in\left\{1,\dots ,m\right\}$;

\item[\emph{(iv)}] $\widetilde{\alpha}_{1}\circ\widetilde{\alpha}_{2}\circ\cdots\circ\widetilde{\alpha}_{m}\left(j\right) =i$ for all locations $j\in\mathcal{N}$.

\end{enumerate}
\end{fixation}
These conditions are similar to those used by \citet{Allen2019} and \citet{McAvoy2021}, except here the Fixation Axiom is modified to account for dynamic networks. Regardless of the configuration of traits, we require the existence of a sequence of replacement maps and networks for the process to transition through, which have non-zero probability (conditions \emph{(i)} and \emph{(ii)}), such that all individuals can trace their ancestry to some $i$ under this sequence (condition \emph{(iii)}), with the additional requirement that this progenitor $i$ is not eternal (condition \emph{(iv)}). One implication is that no individual lives forever and that the process eventually reaches a state in which all individuals are identical by descent, and thus mutant traits can fix within the population (hence ``Fixation'' Axiom). We note that here it does not require each network to be connected.

Since there is no mutation of traits, all individuals must have the same type when they are identical by descent. The configurations $\vA\coloneqq\left(1,1,\dots ,1\right)$ and $\vB\coloneqq\left(0,0,\dots ,0\right)$ are the only absorbing configurations. (Note that while the configuration of types cannot leave $\vA$ or $\vB$, the state itself, which includes the network structure, can still change.) We denote by $\mathbb{B}^{\mathcal{N}}$ the set of all configurations, $\left\{0,1\right\}^{\mathcal{N}}$, and by $\mathbb{B}_{\intercal}^{\mathcal{N}}$ the set of all transient configurations, $\left\{0,1\right\}^{\mathcal{N}}-\left\{\vA ,\vB\right\}$. From the Fixation Axiom, we see that given any starting configuration-network pair, $\left(\vx ,\beta\right)\in\mathbb{B}^{\mathcal{N}}\times\mathcal{L}$, there is a well-defined probability, $\rho_{A}\left(\vx ,\beta\right)$ (resp. $\rho_{B}\left(\vx ,\beta\right)$), that the population eventually reaches the monomorphic state $\vA$ (resp. $\vB$). The behavior of these fixation probabilities (under weak selection, meaning $\delta \ll 1$) is the main focus of this study.

We follow the workflow proposed by \citet{McAvoy2021} for analyzing mutation-free evolutionary dynamics under weak selection. We first study the assortment of traits under neutral drift ($\delta =0$). Subsequently, we link these findings to the game using a martingale perturbation argument. We avoid reproducing the entire derivation in \citep{McAvoy2021}; instead, we highlight the main modifications to those arguments necessary to accommodate stochastic network transitions.

\section{Network-mediated reproductive value}
With the main assumptions in place, we now introduce some derived, demographic quantities that we will refer to throughout the analysis of the model. If the population is in state $\left(\vx ,\beta\right)$, then the marginal probability that $i$ produces an offspring that replaces $j$ in the next update is
\begin{align}
e_{ij}\left(\vx ,\beta\right) &\coloneqq \sum_{\substack{\left(R,\alpha\right) \\ j\in R,\ \alpha\left(j\right) =i}} p_{\left(R,\alpha\right)}\left(\vx ,\beta\right) .
\end{align}
The expected change in the abundance of $A$ in state $\left(\vx ,\beta\right)$ can be expressed as
\begin{align}
\del\left(\vx ,\beta\right) &\coloneqq \sum_{i\in\mathcal{N}} x_{i} \sum_{j\in\mathcal{N}}e_{ij}\left(\vx ,\beta\right) + \sum_{i\in\mathcal{N}}x_{i}\left(1-\sum_{j\in\mathcal{N}}e_{ji}\left(\vx ,\beta\right)\right) -\sum_{i\in\mathcal{N}}x_{i} \nonumber \\
&= \sum_{i,j\in\mathcal{N}} e_{ji}\left(\vx ,\beta\right)\left(x_{j}-x_{i}\right) .
\end{align}

One inconvenient aspect of dealing with the true abundance of $A$ is that it is generally not a martingale under neutral drift. This property is well-known even in models without dynamic structure \citep{Allen2019} and it necessitates working with a weighted frequency instead. The notion of reproductive value, which can be (informally) interpreted as the expected contribution of an individual to future generations, turns out to give the proper weighting. For our purposes, we interpret the reproductive value of $i\in\mathcal{N}$ as the probability that, under neutral drift, $i$ generates a lineage that eventually takes over the population. Because our interest is in fixation probabilities in the first place, it is not surprising that such a quantity should appear. This quantity depends on the network structure, but it is independent of the type configuration due to the drift assumption.

Formally, we define the reproductive value of $i$ in network $\beta$, denoted $\pi_{i}^{\left[\beta\right]}$, to be the probability that under neutral drift and starting in structure $\beta$, a mutant in node $i$ eventually takes over the whole population. Let $e_{ij}^{\circ}\left(\beta\right)$ denote the probability, that under neutral drift and in structure $\beta$, individual $i$ spreads her strategy to $j$. A one-step analysis of the neutral Markov chain gives
\begin{subequations}\label{eq:reproductive_value}
\begin{align}
\pi_{i}^{\left[\beta\right]} &= \sum_{j\in\mathcal{N}}e_{ij}^{\circ}\left(\beta\right) \sum_{\gamma\in\mathcal{L}}q_{\beta\gamma}\pi_{j}^{\left[\gamma\right]}+\left(1-\sum_{j\in \mathcal{N}}e_{ji}^{\circ}\left(\beta\right)\right)\sum_{\gamma\in \mathcal{L}}q_{\beta\gamma}\pi_{i}^{\left[\gamma\right]} ; \label{eq:reproductive_value_recurrence} \\
\sum_{i\in \mathcal{N}}\pi_{i}^{\left[\beta\right]} &= 1
\end{align}
\end{subequations}
for all $i\in\mathcal{N}$ and $\beta\in\mathcal{L}$. There is one point of subtlety in relation to reproductive value on static networks, which relates to the normalization condition $\sum_{i\in\mathcal{N}}\pi_{i}^{\left[\beta\right]}=1$ for all $\beta\in\mathcal{L}$. The Fixation Axiom guarantees that there is a unique $\pi$ satisfying \eq{reproductive_value_recurrence} up to a scalar multiple. In this case, for any fixed $C\in\mathbb{R}$, requiring $\sum_{i\in\mathcal{N}}\sum_{\beta\in\mathcal{L}}\pi_{i}^{\left[\beta\right]}=C$ yields a unique solution to \eq{reproductive_value_recurrence}. Summing both sides of \eq{reproductive_value_recurrence} over $i\in\mathcal{N}$ yields $\sum_{i\in\mathcal{N}}\pi_{i}^{\left[\beta\right]}=\sum_{\gamma\in\mathcal{L}}q_{\beta\gamma}\sum_{i\in\mathcal{N}}\pi_{i}^{\left[\gamma\right]}$. Since the chain $Q$ is irreducible, it follows that $\sum_{i\in\mathcal{N}}\pi_{i}^{\left[\beta\right]}$ is independent of $\beta\in\mathcal{L}$, and thus it must be equal to $C/L$. Therefore, asserting that $\sum_{i\in\mathcal{N}}\sum_{\beta\in\mathcal{L}}\pi_{i}^{\left[\beta\right]}=L$ is equivalent to the requirement that $\sum_{i\in\mathcal{N}}\pi_{i}^{\left[\beta\right]}=1$ for all $\beta\in\mathcal{L}$. As a result, $\pi$, which we refer to as \emph{network-mediated reproductive value} due to its dependence on network transitions, is uniquely defined by \eq{reproductive_value}.

Finally, the change in $\sum_{i\in\mathcal{N}}\pi_{i}^{\left[\beta\right]}x_{i}$, the $\pi$-weighted abundance of $A$, is
\begin{align}
\delhat\left(\vx,\beta\right) &= \sum_{i\in\mathcal{N}} x_{i}\sum_{j\in\mathcal{N}}e_{ij}\left(\vx ,\beta\right) \sum_{\gamma\in\mathcal{L}}q_{\beta\gamma} \pi_{j}^{\left[\gamma\right]} \nonumber \\
&\quad + \sum_{i\in\mathcal{N}}x_{i}\left(1-\sum_{j\in\mathcal{N}}e_{ji}\left(\vx ,\beta\right)\right)\sum_{\gamma\in\mathcal{L}}q_{\beta\gamma}\pi_{i}^{\left[\gamma\right]} -\sum_{i\in\mathcal{N}}\pi_{i}^{\left[\beta\right]}x_{i} \nonumber \\
&= \sum_{i,j\in\mathcal{N}} e_{ji}\left(\vx ,\beta\right) \sum_{\gamma\in\mathcal{L}} q_{\beta\gamma} \pi_{i}^{\left[\gamma\right]} \left(x_{j}-x_{i}\right) + \sum_{i\in\mathcal{N}}x_{i}\left(\sum_{\gamma\in\mathcal{L}}q_{\beta\gamma}\pi_{i}^{\left[\gamma\right]}-\pi_{i}^{\left[\beta\right]}\right) .
\end{align}
It follows from \eq{reproductive_value} that, under neutral drift, $\delhat^{\circ}\left(\vx ,\beta\right) =0$, for all $\vx\in\mathbb{B}^{\mathcal{N}}$ and $\beta\in\mathcal{L}$. This property will play a key role in our subsequent weak-selection analysis of the process (\eq{rhoA_derivative}).

\section{A mutation-modified evolutionary process}
The process under consideration is mutation-free. However, following Ref.~\citep{McAvoy2021}, in order to get an idea of the assortment of types prior to hitting an absorbing configuration, it is convenient to introduce an artificial mutation that makes the chain ergodic and gives it a unique stationary distribution. The idea is to choose a state $\left(\vz ,\lambda\right)$ with $\vz\in\mathbb{B}_{\intercal}^{\mathcal{N}}$, and let mutations bring absorbing configurations into $\left(\vz ,\lambda\right)$ with some small probability $u>0$. If $P_{\left(\vx ,\beta\right)\rightarrow\left(\vy ,\gamma\right)}$ denotes the probability of transitioning from $\left(\vx ,\beta\right)$ to $\left(\vy ,\gamma\right)$ in the original (mutation-free) chain over the course of one time step, then the transition probabilities for the mutation-modified chain are given by
\begin{align} \label{eq:MSS_chain_p}
P_{\left(\vx ,\beta\right)\rightarrow\left(\vy ,\gamma\right)}^{\MSS} &= 
\begin{cases}
u & \vx\in\left\{\vA,\vB\right\} ,\ \left(\vy ,\gamma\right) =\left(\vz ,\lambda\right) , \\
& \\
\left(1-u\right) P_{\left(\vx ,\beta\right)\rightarrow\left(\vy ,\gamma\right)} & \vx\in\left\{\vA ,\vB\right\} ,\ \left(\vy ,\gamma\right)\neq\left(\vz ,\lambda\right) , \\
& \\
P_{\left(\vx ,\beta\right)\rightarrow\left(\vy ,\gamma\right)} & \vx\not\in\left\{\vA, \vB\right\} .
\end{cases}
\end{align}
As a result of the Fixation Axiom, there is a unique stationary distribution, $\pi_{\MSS}$, such that
\begin{align}
\pi_{\MSS}^{\circ}\left(\vx ,\beta\right) &= \sum_{\gamma\in\mathcal{L}} \left(\pi_{\MSS}^{\circ}\left(\vA ,\gamma\right) P_{\left(\vA ,\gamma\right)\rightarrow\left(\vx ,\beta\right)}^{\MSS}+\pi_{\MSS}^{\circ}\left(\vB ,\gamma\right) P_{\left(\vB ,\gamma\right)\rightarrow\left(\vx ,\beta\right)}^{\MSS}\right) \nonumber \\
&\quad + \sum_{\vy\in\mathbb{B}_{\intercal}^{\mathcal{N}}} \sum_{\gamma\in\mathcal{L}} \pi_{\MSS}^{\circ}\left(\vy ,\gamma\right) P_{\left(\vy ,\gamma\right)\rightarrow\left(\vx ,\beta\right)}^{\MSS} \nonumber \\
&= \sum_{\gamma\in\mathcal{L}} \pi_{\MSS}^{\circ}\left(\vA ,\gamma\right) \left( u\delta_{\vz ,\vx}\delta_{\lambda ,\beta} +\left(1-u\right) \delta_{\vA ,\vx}q_{\gamma\beta}\right) \nonumber \\
&\quad +\sum_{\gamma\in\mathcal{L}} \pi_{\MSS}^{\circ}\left(\vB ,\gamma\right)\left( u\delta_{\vz ,\vx}\delta_{\lambda ,\beta} +\left(1-u\right) \delta_{\vB ,\vx}q_{\gamma\beta}\right) \nonumber \\
&\quad + \sum_{\vy\in\mathbb{B}_{\intercal}^{\mathcal{N}}} \sum_{\gamma\in\mathcal{L}} \pi_{\MSS}^{\circ}\left(\vy ,\gamma\right) P_{\left(\vy ,\gamma\right)\rightarrow\left(\vx ,\beta\right)} \label{eq:MSS_recurrence}
\end{align}
for all $\vx\in\mathbb{B}$ and $\beta\in\mathcal{L}$.

In one step after state $\left(\vx,\beta\right)$, the expected change in the $\pi$-weighted abundance of $A$ is 
\begin{align}
\delhat_{\MSS}\left(\vx,\beta\right) &= 
\begin{cases}
-u\left(1-\sum_{i\in\mathcal{N}}\pi_{i}^{\left[\lambda\right]}z_{i}\right) & \vx =\vA , \\
& \\
u\sum_{i\in\mathcal{N}}\pi_{i}^{\left[\lambda\right]}z_{i} & \vx =\vB , \\
& \\
\delhat\left(\vx,\beta\right) & \vx\not\in\left\{\vA ,\vB\right\} .
\end{cases}
\end{align}
Averaging this expected change over the stationary distribution of the modified chain gives
\begin{align}
0 &= \E_{\MSS}\left[\delhat_{\MSS}\right] \nonumber \\
&= \E_{\MSS}\left[\delhat\right] -u\sum_{\beta\in \mathcal{L}}\pi_{\MSS}\left(\vA,\beta\right)\left(1-\sum_{i\in\mathcal{N}}\pi_{i}^{\left[\lambda\right]}z_{i}\right) \nonumber \\
&\quad +u\sum_{\beta\in \mathcal{L}}\pi_{\MSS}\left(\vB,\beta\right)\sum_{i\in\mathcal{N}}\pi_{i}^{\left[\lambda\right]}z_{i} . \label{eq:MSS_chain_Delta}
\end{align}
Owing to a result of \citet{fudenberg:JET:2006}, we know that, in the low-mutation limit,
\begin{subequations}
\begin{align}
\lim_{u\rightarrow 0}\sum_{\beta\in\mathcal{L}}\pi_{\MSS}\left(\vA ,\beta\right) &= \rho_{A}\left(\vz ,\lambda\right) ; \\
\lim_{u\rightarrow 0}\sum_{\beta\in\mathcal{L}}\pi_{\MSS}\left(\vB ,\beta\right) &= \rho_{B}\left(\vz ,\lambda\right) .
\end{align}
\end{subequations}
Therefore, taking the derivative of both sides of \eq{MSS_chain_Delta} with respect to $u$ at $u=0$ gives
\begin{align}
\rho_{A}\left(\vz ,\lambda\right) &= \sum_{i\in\mathcal{N}}\pi_{i}^{\left[\lambda\right]}z_{i}+\frac{d}{du}\Bigg\vert_{u=0}\mathbb{E}_{\MSS}\left[\delhat\right] .
\end{align}

Let $\left<\cdot\right>_{\left(\vz ,\lambda\right)}\coloneqq\frac{d}{du}\Big\vert_{u=0}\mathbb{E}_{\MSS}\left[\cdot\right]$. By the argument given in Ref.~\citep[][Corollary 1]{McAvoy2021}, we see that for any function $\varphi :\mathbb{B}^{\mathcal{N}}\times\mathcal{L}\rightarrow\mathbb{R}$ satisfying $\varphi\left(\vA ,\beta\right) =\varphi\left(\vB ,\beta\right) =0$ for all $\beta\in\mathcal{L}$,
\begin{align}
\left<\varphi\right>_{\left(\vz ,\lambda\right)} &= \sum_{t=0}^{\infty} \mathbb{E}\left[ \varphi\left(\vx^{t},\beta^{t}\right)\mid\left(\vx^{0},\beta^{0}\right) =\left(\vz ,\lambda\right)\right] , \label{eq:MSS_derivative_sum}
\end{align}
where the summation on the right-hand side converges absolutely. In particular, this equation holds for the expected change in the $\pi$-weighted abundance of $A$, $\varphi =\delhat$. Since we also have
\begin{align}
\frac{d}{d\delta}\Bigg\vert_{\delta =0} e_{ij}\left(\vx,\beta\right) &= \sum_{I\subseteq \mathcal{N}} c_{I}^{ij}\left(\beta\right) \vx_{I} \label{eq:eji}
\end{align}
for unique coefficients $c_{I}^{ij}\left(\beta\right)$, where $\vx_{I}\coloneqq\prod_{i\in I}x_{i}$, it follows that
\begin{align}
\frac{d}{d\delta}\Bigg\vert_{\delta =0}\rho_{A}\left(\vz ,\lambda\right) &= \frac{d}{d\delta}\Bigg\vert_{\delta =0}\left<\delhat\right>_{\left(\vz ,\lambda\right)} \nonumber \\
&= \left<\frac{d}{d\delta}\Bigg\vert_{\delta =0}\delhat\right>_{\left(\vz ,\lambda\right)}^{\circ} \nonumber \\
&= \left<\frac{d}{d\delta}\Bigg\vert_{\delta =0}\sum_{i,j\in\mathcal{N}} e_{ji}\left(\vx ,\beta\right) \sum_{\gamma\in\mathcal{L}} q_{\beta\gamma} \pi_{i}^{\left[\gamma\right]} \left(x_{j}-x_{i}\right)\right>_{\left(\vz ,\lambda\right)}^{\circ} \nonumber \\
&= \sum_{i,j\in\mathcal{N}} \sum_{I\subseteq\mathcal{N}} \left<c_{I}^{ji}\left(\beta\right) \sum_{\gamma\in\mathcal{L}} q_{\beta\gamma} \pi_{i}^{\left[\gamma\right]} \left(\vx_{I\cup\left\{j\right\}}-\vx_{I\cup\left\{i\right\}}\right)\right>_{\left(\vz ,\lambda\right)}^{\circ} , \label{eq:rhoA_derivative}
\end{align}
where the interchange of the two limits is possible due to \eq{MSS_derivative_sum} and the absolute convergence of its summation. The second line of \eq{rhoA_derivative} is where we use the fact that $\delhat^{0}\left(\vx ,\beta\right) =0$ for all $\vx\in\mathbb{B}^{\mathcal{N}}$ and $\beta\in\mathcal{L}$, highlighting the importance of network-mediated reproductive value.

As a result of these calculations, what remains in order to understand the first-order effects of selection on a mutant type's fixation probability is an analysis of the neutral operator $\left<\cdot\right>_{\left(\vz ,\lambda\right)}^{\circ}$.

\section{Analysis of neutral drift}
Throughout this section, we denote the stationary distribution of the structure-transition chain, $Q$, by $\upsilon$. We also suppress either the configuration or the network when we marginalize. For example, we write $\pi_{\MSS}\left(\vx\right)$ for $\sum_{\beta\in\mathcal{L}}\pi_{\MSS}\left(\vx ,\beta\right)$ and $\pi_{\MSS}\left(\beta\right)$ for $\sum_{\vx\in\mathbb{B}^{\mathcal{N}}}\pi_{\MSS}\left(\vx ,\beta\right)$.

In the limit of low mutation, we know $\pi_{\MSS}^{\circ}\left(\vA\right)$ converges to $\rho_{A}^{\circ}\left(\vz ,\lambda\right)$ and $\pi_{\MSS}^{\circ}\left(\vB\right)$ converges to $\rho_{B}^{\circ}\left(\vz ,\lambda\right)$. The following lemma is a slightly stronger version of this result:
\begin{lemma}\label{lemma:rho_low_u}
For all networks $\beta\in\mathcal{L}$,
\begin{subequations}
\begin{align}
\lim_{u\rightarrow 0} \pi_{\MSS}^{\circ}\left(\vA ,\beta\right) &= \rho_{A}^{\circ}\left(\vz ,\lambda\right)\upsilon\left(\beta\right) ; \\
\lim_{u\rightarrow 0} \pi_{\MSS}^{\circ}\left(\vB ,\beta\right) &= \rho_{B}^{\circ}\left(\vz ,\lambda\right)\upsilon\left(\beta\right) .
\end{align}
\end{subequations}
\end{lemma}
\begin{proof}
Letting $\vx =\vA$ in \eq{MSS_recurrence} and taking $u\rightarrow 0$ gives
\begin{align}
\lim_{u\rightarrow 0}\pi_{\MSS}^{\circ}\left(\vA ,\beta\right) &= \sum_{\gamma\in\mathcal{L}} \left( \lim_{u\rightarrow 0}\pi_{\MSS}^{\circ}\left(\vA ,\gamma\right) \right) q_{\gamma\beta} .
\end{align}
It follows that $\lim_{u\rightarrow 0}\pi_{\MSS}^{\circ}\left(\vA ,\beta\right)$ is proportional to $\upsilon\left(\beta\right)$, for all $\beta\in\mathcal{L}$. The constant of proportionality must be $\rho_{A}^{\circ}\left(\vz ,\lambda\right)$ due to the fact that $\lim_{u\rightarrow 0}\pi_{\MSS}^{\circ}\left(\vA\right) =\rho_{A}^{\circ}\left(\vz ,\lambda\right)$. The result for $\lim_{u\rightarrow 0} \pi_{\MSS}^{\circ}\left(\vB ,\beta\right)$ follows from analogous reasoning and is omitted here.
\end{proof}

\begin{remark}
Neutral fixation probabilities, $\rho_{A}^{\circ}\left(\vz ,\lambda\right)$ and $\rho_{B}^{\circ}\left(\vz ,\lambda\right)$, can be calculated using reproductive values and the identities $\rho_{A}^{\circ}\left(\vz ,\lambda\right) =\sum_{i\in\mathcal{N}}\pi_{i}^{\left[\lambda\right]}z_{i}$ and $\rho_{B}^{\circ}\left(\vz ,\lambda\right) =1-\sum_{i\in\mathcal{N}}\pi_{i}^{\left[\lambda\right]}z_{i}$.
\end{remark}

The following is an immediate consequence of Lemma~\ref{lemma:rho_low_u}:
\begin{corollary}
$\lim_{u\rightarrow 0}\pi_{\MSS}^{\circ}\left(\beta\right) =\upsilon\left(\beta\right)$.
\end{corollary}

The next lemma establishes a recurrence for $\frac{d}{du}\Big\vert_{u=0}\pi_{\MSS}^{\circ}\left(\beta\right)$:
\begin{lemma}\label{lemma:du_pi_beta}
For every $\beta$, we have
\begin{align}
\frac{d}{du}\Bigg\vert_{u=0}\pi_{\MSS}^{\circ}\left(\beta\right) &= \delta_{\beta ,\lambda}-\upsilon\left(\beta\right) + \sum_{\gamma\in\mathcal{L}} \left(\frac{d}{du}\Bigg\vert_{u=0} \pi_{\MSS}^{\circ}\left(\gamma\right)\right) q_{\gamma\beta} . \label{eq:structure_derivative_recurrence}
\end{align}
\end{lemma}
\begin{proof}
Summing both sides of \eq{MSS_recurrence} over all $\vx\in\mathbb{B}^{\mathcal{N}}$ gives
\begin{align}
\pi_{\MSS}^{\circ}\left(\beta\right) &= u\sum_{\gamma\in\mathcal{L}} \left(\pi_{\MSS}^{\circ}\left(\vA ,\gamma\right) +\pi_{\MSS}^{\circ}\left(\vB ,\gamma\right) \right)\left(\delta_{\beta ,\lambda}-q_{\gamma\beta}\right) \nonumber \\
&\quad + \sum_{\gamma\in\mathcal{L}} \pi_{\MSS}^{\circ}\left(\gamma\right) q_{\gamma\beta} .
\end{align}
Differentiating this equation with respect to $u$ at $u=0$ and using Lemma~\ref{lemma:rho_low_u} yields \eq{structure_derivative_recurrence}.
\end{proof}

Since the state of the process consists of both a configuration of traits and a network structure, the next result gives a recurrence for calculating a modified version of $\left<\cdot\right>_{\left(\vz ,\lambda\right)}^{\circ}$, using conditioning on the network structure. In particular, for a function $\varphi :\mathbb{B}^{\mathcal{N}}\rightarrow\mathbb{R}$ defined on \emph{just} configurations, we let $\left<\varphi\mid\beta\right>_{\left(\vz ,\lambda\right)}^{\circ}=\frac{d}{du}\Big\vert_{u=0}\mathbb{E}_{\MSS}^{\circ}\left[\varphi\mid\beta\right]$. This quantity can be calculated as follows:
\begin{proposition}\label{prop:conditional_recurrence}
For every function $\varphi :\mathbb{B}^{\mathcal{N}}\rightarrow\mathbb{R}$, we have
\begin{align}
\upsilon\left(\beta\right) \left<\varphi\mid\beta\right>_{\left(\vz ,\lambda\right)}^{\circ} &= \delta_{\lambda ,\beta}\left(\varphi\left(\vz\right) -\rho_{A}^{\circ}\left(\vz ,\lambda\right)\varphi\left(\vA\right) -\rho_{B}^{\circ}\left(\vz ,\lambda\right)\varphi\left(\vB\right)\right) \nonumber \\
&\quad + \sum_{\gamma\in\mathcal{L}} \upsilon\left(\gamma\right) \sum_{\left(R,\alpha\right)} p_{\left(R,\alpha\right)}^{\circ}\left(\gamma\right) q_{\gamma\beta} \left<\varphi_{\widetilde{\alpha}}\mid\gamma\right>_{\left(\vz ,\lambda\right)}^{\circ} , \label{eq:conditional_recurrence}
\end{align}
where, for $\widetilde{\alpha}:\mathcal{N}\rightarrow\mathcal{N}$, $\varphi_{\widetilde{\alpha}}:\mathbb{B}^{\mathcal{N}}\rightarrow\mathbb{R}$ is the map defined by $\varphi_{\widetilde{\alpha}}\left(\vx\right) =\varphi\left(\vx_{\widetilde{\alpha}}\right)$ for $\vx\in\mathbb{B}^{\mathcal{N}}$.
\end{proposition}
\begin{proof}
For $\vx\in\mathbb{B}_{\intercal}^{\mathcal{N}}$, differentiating both sides of \eq{MSS_recurrence} with respect to $u$ at $u=0$ gives
\begin{align}
\frac{d}{du} &\Bigg\vert_{u=0} \pi_{\MSS}^{\circ}\left(\vx ,\beta\right) \nonumber \\
&= \delta_{\vz ,\vx} \delta_{\lambda ,\beta} + \sum_{\vy\in\mathbb{B}_{\intercal}^{\mathcal{N}}} \sum_{\gamma\in\mathcal{L}} \left(\frac{d}{du}\Bigg\vert_{u=0} \pi_{\MSS}^{\circ}\left(\vy ,\gamma\right)\right) P_{\left(\vy ,\gamma\right)\rightarrow\left(\vx ,\beta\right)}^{\circ} \nonumber \\
&= \delta_{\vz ,\vx} \delta_{\lambda ,\beta} + \sum_{\vy\in\mathbb{B}_{\intercal}^{\mathcal{N}}} \sum_{\gamma\in\mathcal{L}} \left(\frac{d}{du}\Bigg\vert_{u=0} \pi_{\MSS}^{\circ}\left(\vy ,\gamma\right)\right) \sum_{\substack{\left(R,\alpha\right) \\ \vy_{\widetilde{\alpha}}=\vx}} p_{\left(R,\alpha\right)}^{\circ}\left(\gamma\right) q_{\gamma\beta} . \label{eq:dpi_transient}
\end{align}
Doing so for $\vx\in\left\{\vA ,\vB\right\}$ gives
\begin{subequations}\label{eq:dpi_absorbing}
\begin{align}
\frac{d}{du}\Bigg\vert_{u=0} \pi_{\MSS}^{\circ}\left(\vA ,\beta\right) &= \sum_{\gamma\in\mathcal{L}} \left(\frac{d}{du}\Bigg\vert_{u=0}\pi_{\MSS}^{\circ}\left(\vA ,\gamma\right)\right) q_{\gamma\beta} - \rho_{A}^{\circ}\left(\vz ,\lambda\right)\upsilon\left(\beta\right) \nonumber \\
&\quad + \sum_{\vy\in\mathbb{B}_{\intercal}^{\mathcal{N}}} \sum_{\gamma\in\mathcal{L}} \left(\frac{d}{du}\Bigg\vert_{u=0} \pi_{\MSS}^{\circ}\left(\vy ,\gamma\right)\right) \sum_{\substack{\left(R,\alpha\right) \\ \vy_{\widetilde{\alpha}}=\vA}} p_{\left(R,\alpha\right)}^{\circ}\left(\gamma\right) q_{\gamma\beta} ; \\
\frac{d}{du}\Bigg\vert_{u=0} \pi_{\MSS}^{\circ}\left(\vB ,\beta\right) &= \sum_{\gamma\in\mathcal{L}} \left(\frac{d}{du}\Bigg\vert_{u=0}\pi_{\MSS}^{\circ}\left(\vB ,\gamma\right)\right) q_{\gamma\beta} - \rho_{B}^{\circ}\left(\vz ,\lambda\right)\upsilon\left(\beta\right) \nonumber \\
&\quad +\sum_{\vy\in\mathbb{B}_{\intercal}^{\mathcal{N}}} \sum_{\gamma\in\mathcal{L}} \left(\frac{d}{du}\Bigg\vert_{u=0} \pi_{\MSS}^{\circ}\left(\vy ,\gamma\right)\right) \sum_{\substack{\left(R,\alpha\right) \\ \vy_{\widetilde{\alpha}}=\vB}} p_{\left(R,\alpha\right)}^{\circ}\left(\gamma\right) q_{\gamma\beta} .
\end{align}
\end{subequations}

If $\varphi :\mathbb{B}^{\mathcal{N}}\rightarrow\mathbb{R}$ is a fixed function, then, by definition,
\begin{align}
\upsilon\left(\beta\right)\left<\varphi\mid\beta\right>_{\left(\vz ,\lambda\right)}^{\circ} &= \sum_{\vx\in\mathbb{B}^{\mathcal{N}}} \upsilon\left(\beta\right)\frac{d}{du}\Bigg\vert_{u=0} \frac{\pi_{\MSS}^{\circ}\left(\vx ,\beta\right)}{\pi_{\MSS}^{\circ}\left(\beta\right)} \varphi\left(\vx\right) .
\end{align}
Combining Lemma~\ref{lemma:du_pi_beta} and \eqs{dpi_transient}{dpi_absorbing} with the fact that
\begin{align}
\upsilon\left(\beta\right) &\frac{d}{du}\Bigg\vert_{u=0} \frac{\pi_{\MSS}^{\circ}\left(\vx ,\beta\right)}{\pi_{\MSS}^{\circ}\left(\beta\right)} \nonumber \\
&= \frac{d}{du}\Bigg\vert_{u=0} \pi_{\MSS}^{\circ}\left(\vx ,\beta\right) - \left(\delta_{\vA ,\vx}\rho_{A}^{\circ}\left(\vz ,\lambda\right) +\delta_{\vB ,\vx}\rho_{B}^{\circ}\left(\vz ,\lambda\right)\right)\frac{d}{du}\Bigg\vert_{u=0}\pi_{\MSS}^{\circ}\left(\beta\right)
\end{align}
then gives \eq{conditional_recurrence} after some tedious but straightforward simplifications.
\end{proof}

\begin{corollary}\label{cor:conditional_recurrence}
With $I\subseteq\mathcal{N}$ and $\eta_{I}^{\left[\beta\right]}\left(\vz ,\lambda\right)\coloneqq\upsilon\left(\beta\right)\left<\sum_{i\in\mathcal{N}}\pi_{i}^{\left[\beta\right]}x_{i}-\vx_{I}\mid\beta\right>_{\left(\vz ,\lambda\right)}^{\circ}$, we have
\begin{align}
\eta_{I}^{\left[\beta\right]}\left(\vz ,\lambda\right) &= \delta_{\lambda ,\beta}\left(\sum_{i\in\mathcal{N}}\pi_{i}^{\left[\beta\right]}z_{i}-\vz_{I}\right) + \sum_{\gamma\in\mathcal{L}}\sum_{\left(R,\alpha\right)}p_{\left(R,\alpha\right)}^{\circ}\left(\gamma\right) q_{\gamma\beta}\eta_{\widetilde{\alpha}\left(I\right)}^{\left[\gamma\right]}\left(\vz ,\lambda\right) . \label{eq:eta_recurrence}
\end{align}
Subject to $\sum_{i\in\mathcal{N}}\pi_{i}^{\left[\beta\right]}\eta_{i}^{\left[\beta\right]}\left(\vz ,\lambda\right) =0$ for some $\beta\in\mathcal{L}$, the solution to \eq{eta_recurrence} is unique.
\end{corollary}
\begin{proof}
Setting $\varphi\left(\vx\right) =\sum_{i\in\mathcal{N}}\pi_{i}^{\left[\beta\right]}x_{i}-\vx_{I}$ in Proposition~\ref{prop:conditional_recurrence} gives \eq{eta_recurrence}. Conversely, we know that $\eta_{I}^{\left[\beta\right]}\left(\vz ,\lambda\right)\coloneqq\upsilon\left(\beta\right)\left<\sum_{i\in\mathcal{N}}\pi_{i}^{\left[\beta\right]}x_{i}-\vx_{I}\mid\beta\right>_{\left(\vz ,\lambda\right)}^{\circ}$ solves \eq{eta_recurrence}, so that there is at least one solution to \eq{eta_recurrence}. By the Fixation Axiom, the dimensionality of the space of  solutions to \eq{eta_recurrence} is determined by that of the case $\left| I\right| =1$. (The reason is that all subsets of size greater than one are transient under the ancestral process.) Specifically, the recurrence for $I=\left\{i\right\}$ is
\begin{align}
\eta_{i}^{\left[\beta\right]}\left(\vz ,\lambda\right) &= \delta_{\lambda ,\beta}\left(\rho_{A}^{\circ}\left(\vz ,\lambda\right)-z_{i}\right) + \sum_{\gamma\in\mathcal{L}}\sum_{j\in\mathcal{N}} e_{ji}^{\circ}\left(\gamma\right) q_{\gamma\beta}\eta_{j}^{\left[\gamma\right]}\left(\vz ,\lambda\right) \nonumber \\
&\quad + \sum_{\gamma\in\mathcal{L}}\left(1-\sum_{j\in\mathcal{N}}e_{ji}^{\circ}\left(\gamma\right)\right) q_{\gamma\beta}\eta_{i}^{\left[\gamma\right]}\left(\vz ,\lambda\right) . \label{eq:eta_recurrence_singleton}
\end{align}
If $\widetilde{\eta}\left(\vz ,\lambda\right)$ is another solution to \eq{eta_recurrence_singleton}, then $\chi\left(\vz ,\lambda\right)\coloneqq\eta\left(\vz ,\lambda\right) -\widetilde{\eta}\left(\vz ,\lambda\right)$ satisfies
\begin{align}
\chi_{i}^{\left[\beta\right]}\left(\vz ,\lambda\right) &= \sum_{\gamma\in\mathcal{L}}\sum_{j\in\mathcal{N}} e_{ji}^{\circ}\left(\gamma\right) q_{\gamma\beta}\chi_{j}^{\left[\gamma\right]}\left(\vz ,\lambda\right) + \sum_{\gamma\in\mathcal{L}}\left(1-\sum_{j\in\mathcal{N}}e_{ji}^{\circ}\left(\gamma\right)\right) q_{\gamma\beta}\chi_{i}^{\left[\gamma\right]}\left(\vz ,\lambda\right) . \label{eq:chi_recurrence}
\end{align}
Noting that any constant function is a solution to \eq{chi_recurrence}, and the space of solutions to this equation is one-dimensional as a result of the Fixation Axiom, there must exist $K\in\mathbb{R}$ such that $\eta\left(\vz ,\lambda\right) = \widetilde{\eta}\left(\vz ,\lambda\right)+K$. Since the solution $\eta_{i}^{\left[\beta\right]}\left(\vz ,\lambda\right)=\upsilon\left(\beta\right)\left<x_{i}\mid\beta\right>_{\left(\vz ,\lambda\right)}^{\circ}$ satisfies $\sum_{i\in\mathcal{N}}\pi_{i}^{\left[\beta\right]}\eta_{i}^{\left[\beta\right]}\left(\vz ,\lambda\right)=0$ for all $\beta\in\mathcal{L}$, it follows that $K=0$ and $\eta\left(\vz ,\lambda\right) =\widetilde{\eta}\left(\vz ,\lambda\right)$ whenever $\widetilde{\eta}\left(\vz ,\lambda\right)$ satisfies \eq{eta_recurrence} and $\sum_{i\in\mathcal{N}}\eta_{i}^{\left[\beta\right]}\widetilde{\eta}_{i}^{\left[\beta\right]}\left(\vz ,\lambda\right) =0$ for some $\beta\in\mathcal{L}$. We note that $\sum_{i\in\mathcal{N}}\pi_{i}^{\left[\beta\right]}\eta_{i}^{\left[\beta\right]}\left(\vz ,\lambda\right)=0$ for \emph{some} $\beta\in\mathcal{L}$ ensures that this equation holds for \emph{all} $\beta\in\mathcal{L}$.
\end{proof}

\section{Calculating first-order effects of selection}\label{sec:calculating_first_order}

\subsection{Fixed initial configurations}\label{sec:fixed_configurations}
Note that for functions $\varphi :\mathbb{B}^{\mathcal{N}}\rightarrow\mathbb{R}$ and $\phi :\mathcal{L}\rightarrow\mathbb{R}$, we have
\begin{align}
\left<\phi\varphi\right>_{\left(\vz ,\lambda\right)}^{\circ} &= \frac{d}{du}\Bigg\vert_{u=0} \sum_{\beta\in\mathcal{L}}\pi_{\MSS}^{\circ}\left(\beta\right) \phi\left(\beta\right) \mathbb{E}_{\MSS}^{\circ}\left[ \varphi\mid\beta \right] \nonumber \\
&= \sum_{\beta\in\mathcal{L}}\upsilon\left(\beta\right) \phi\left(\beta\right) \left<\varphi\mid\beta \right>_{\left(\vz ,\lambda\right)}^{\circ} \nonumber \\
&\quad + \left(\rho_{A}^{\circ}\left(\vz ,\lambda\right)\varphi\left(\vA\right) +\rho_{B}^{\circ}\left(\vz ,\lambda\right)\varphi\left(\vB\right)\right) \sum_{\beta\in\mathcal{L}}\phi\left(\beta\right) \frac{d}{du}\Bigg\vert_{u=0}\pi_{\MSS}^{\circ}\left(\beta\right) .
\end{align}
Therefore, we may rewrite \eq{rhoA_derivative} as
\begin{align}
\frac{d}{d\delta} &\Bigg\vert_{\delta =0}\rho_{A}\left(\vz ,\lambda\right) \nonumber \\
&= \sum_{i,j\in\mathcal{N}} \sum_{I\subseteq\mathcal{N}} \left<\underbrace{c_{I}^{ji}\left(\beta\right) \textstyle\sum_{\gamma\in\mathcal{L}} q_{\beta\gamma} \pi_{i}^{\left[\gamma\right]}}_{\phi} \Big(\underbrace{\vx_{I\cup\left\{j\right\}}-\vx_{I\cup\left\{i\right\}}}_{\varphi}\Big)\right>_{\left(\vz ,\lambda\right)}^{\circ} \nonumber \\
&= \sum_{i,j\in\mathcal{N}} \sum_{I\subseteq\mathcal{N}} \sum_{\beta\in\mathcal{L}} \upsilon\left(\beta\right) c_{I}^{ji}\left(\beta\right) \sum_{\gamma\in\mathcal{L}} q_{\beta\gamma} \pi_{i}^{\left[\gamma\right]}\left( \left<\vx_{I\cup\left\{j\right\}}\mid\beta\right>_{\left(\vz ,\lambda\right)}^{\circ} - \left<\vx_{I\cup\left\{i\right\}}\mid\beta\right>_{\left(\vz ,\lambda\right)}^{\circ}\right) .
\end{align}
Defining $\eta_{I}^{\left[\beta\right]}\left(\vz,\lambda\right)\coloneqq\upsilon\left(\beta\right)\left<\sum_{i\in\mathcal{N}}\pi_{i}^{\left[\beta\right]}x_{i}-\vx_{I}\mid\beta\right>_{\left(\vz ,\lambda\right)}^{\circ}$, we then have
\begin{align}
\frac{d}{d\delta}\Bigg\vert_{\delta =0}\rho_{A}\left(\vz ,\lambda\right) &= \sum_{i,j\in\mathcal{N}} \sum_{I\subseteq\mathcal{N}} \sum_{\beta\in\mathcal{L}} c_{I}^{ji}\left(\beta\right) \sum_{\gamma\in\mathcal{L}} q_{\beta\gamma} \pi_{i}^{\left[\gamma\right]}\left( \eta_{I\cup\left\{i\right\}}^{\left[\beta\right]}\left(\vz ,\lambda\right) - \eta_{I\cup\left\{j\right\}}^{\left[\beta\right]}\left(\vz ,\lambda\right)\right) ,
\end{align}
where, by Corollary~\ref{cor:conditional_recurrence}, the terms $\eta$ are uniquely determined by
\begin{subequations}
\begin{align}
\eta_{I}^{\left[\beta\right]}\left(\vz ,\lambda\right) &= \delta_{\lambda ,\beta}\left(\sum_{i\in\mathcal{N}}\pi_{i}^{\left[\beta\right]}z_{i}-\vz_{I}\right) + \sum_{\gamma\in\mathcal{L}}\sum_{\left(R,\alpha\right)}p_{\left(R,\alpha\right)}^{\circ}\left(\gamma\right) q_{\gamma\beta}\eta_{\widetilde{\alpha}\left(I\right)}^{\left[\gamma\right]}\left(\vz ,\lambda\right) ; \\
\sum_{i\in\mathcal{N}} \pi_{i}^{\left[\beta\right]} \eta_{i}^{\left[\beta\right]}\left(\vz ,\lambda\right) &= 0 \textrm{ for some }\beta\in\mathcal{L} .
\end{align}
\end{subequations}

\subsection{Probabilistic initial configurations}
Up until this point, we have focused on fixation probabilities given some fixed initial state, $\left(\vz ,\lambda\right)\in\mathcal{N}\times\mathcal{L}$. We now allow mutant types to arise stochastically and consider mean fixation probabilities for both types. For two distributions, $\mu_{A},\mu_{B}\in\Delta\left(\mathbb{B}_{\intercal}^{\mathcal{N}}\times\mathcal{L}\right)$, we let
\begin{subequations}
\begin{align}
\rho_{A}\left(\mu_{A}\right) &\coloneqq \mathbb{E}_{\left(\vz ,\lambda\right)\sim\mu_{A}}\left[\rho_{A}\left(\vz ,\lambda\right)\right] ; \\
\rho_{B}\left(\mu_{B}\right) &\coloneqq \mathbb{E}_{\left(\vz ,\lambda\right)\sim\mu_{B}}\left[\rho_{B}\left(\vz ,\lambda\right)\right] .
\end{align}
\end{subequations}
By the results of \S\ref{sec:fixed_configurations}, for any $\mu\in\Delta\left(\mathbb{B}_{\intercal}^{\mathcal{N}}\times\mathcal{L}\right)$, we have
\begin{align}
\frac{d}{d\delta}\Bigg\vert_{\delta =0}\rho_{A}\left(\mu\right) &= \sum_{i,j\in\mathcal{N}} \sum_{I\subseteq\mathcal{N}} \sum_{\beta\in\mathcal{L}} c_{I}^{ji}\left(\beta\right) \sum_{\gamma\in\mathcal{L}} q_{\beta\gamma} \pi_{i}^{\left[\gamma\right]}\left( \eta_{I\cup\left\{i\right\}}^{\left[\beta\right]}\left(\mu\right) - \eta_{I\cup\left\{j\right\}}^{\left[\beta\right]}\left(\mu\right)\right) , \label{eq:rhoA_probabilistic}
\end{align}
where
\begin{subequations}\label{eq:eta_mean}
\begin{align}
\eta_{I}^{\left[\beta\right]}\left(\mu\right) &= \mathbb{E}_{\left(\vz ,\lambda\right)\sim\mu}\left[\delta_{\lambda ,\beta}\left(\sum_{i\in\mathcal{N}}\pi_{i}^{\left[\beta\right]}z_{i}-\vz_{I}\right)\right] \nonumber \\
&\quad + \sum_{\gamma\in\mathcal{L}}\sum_{\left(R,\alpha\right)}p_{\left(R,\alpha\right)}^{\circ}\left(\gamma\right) q_{\gamma\beta}\eta_{\widetilde{\alpha}\left(I\right)}^{\left[\gamma\right]}\left(\mu\right) ; \\
\sum_{i\in\mathcal{N}} \pi_{i}^{\left[\beta\right]} \eta_{i}^{\left[\beta\right]}\left(\mu\right) &= 0 \textrm{ for some }\beta\in\mathcal{L} .
\end{align}
\end{subequations}
Letting $\mu =\mu_{A}$ gives the mean fixation probability for type $A$, while the mean fixation probability for type $B$ can be calculated analogously using the equation $\rho_{B}\left(\mu_{B}\right) =1-\rho_{A}\left(\mu_{B}\right)$.

Although the main focus of our study is on network-transition chains that are both aperiodic and irreducible, we do also consider periodic structures. Suppose that among the $L$ networks in $\mathcal{L}$, network $\beta$ transitions deterministically to network $\beta +1$ for $\beta\in\left\{1,\dots ,L-1\right\}$, and network $L$ transitions deterministically to network $1$. We can then write \eq{eta_mean} more explicitly as
\begin{subequations}\label{eq:periodic}
\begin{align}
\eta_{I}^{\left[1\right]}\left(\mu\right) &= \mathbb{E}_{\left(\vz ,\lambda\right)\sim\mu}\left[\delta_{\lambda ,1}\left(\sum_{i\in\mathcal{N}}\pi_{i}^{\left[1\right]}z_{i}-\vz_{I}\right)\right] + \sum_{\left(R,\alpha\right)}p_{\left(R,\alpha\right)}^{\circ}\left(L\right) \eta_{\widetilde{\alpha}\left(I\right)}^{\left[L\right]}\left(\mu\right) ; \\
\eta_{I}^{\left[\beta\right]}\left(\mu\right) &= \mathbb{E}_{\left(\vz ,\lambda\right)\sim\mu}\left[\delta_{\lambda ,\beta}\left(\sum_{i\in\mathcal{N}}\pi_{i}^{\left[\beta\right]}z_{i}-\vz_{I}\right)\right] \nonumber \\
&\quad + \sum_{\left(R,\alpha\right)}p_{\left(R,\alpha\right)}^{\circ}\left(\beta -1\right) \eta_{\widetilde{\alpha}\left(I\right)}^{\left[\beta -1\right]}\left(\mu\right) ; \quad \left( 1<\beta\leqslant L\right) \\
\sum_{i\in\mathcal{N}} \pi_{i}^{\left[\beta\right]} \eta_{i}^{\left[\beta\right]}\left(\mu\right) &= 0 \textrm{ for some }\beta\in\mathcal{L} .
\end{align}
\end{subequations}

\subsection{Changing population size}
Although our modeling framework is described for transitions between networks of a fixed size, it can also accommodate populations of changing size. Indeed, the size of the networks, $N$, may be thought of as the size of an ambient space in which not all nodes are active at any given time. It need not be true that all networks are connected for the Fixation Axiom to hold. In particular, network $\beta$ could have $N^{\left[\beta\right]}$ non-isolated nodes (representing a graph of size $N^{\left[\beta\right]}\leqslant N$) and $N-N^{\left[\beta\right]}$ isolated nodes. The isolated nodes represent a padding of the actual population so that the resulting model fits within the framework of structural transitions among networks of size $N$.

For example, under death-birth updating, a non-isolated node $i$ in network $\beta$ is chosen with probability $1/N^{\left[\beta\right]}$. All neighboring nodes of $i$ in $\beta$ then compete to reproduce, with probability proportional to edge-weight-weighted reproductive rate. Since $i$ is non-isolated, there is at least one such neighbor, and the offspring is sent to $i$. If $p_{ij}^{\left[\beta\right]}=w_{ji}^{\left[\beta\right]}/\sum_{\ell=1}^N w_{\ell i}^{\left[\beta\right]}$ when $i$ is a non-isolated node and $p_{ij}^{\left[\beta\right]}=0$ otherwise, then the probability that $j$ transmits its offspring to $i$ is
\begin{equation}
e_{ji}\left(\mathbf{x},\beta\right) = \frac{1}{N^{\left[\beta\right]}}\frac{p_{ij}^{\left[\beta\right]}F_{j}\left(\mathbf{x},\beta\right)}{\sum_{\ell=1}^{N} p_{i\ell}^{\left[\beta\right]}F_{\ell}\left(\mathbf{x},\beta\right)} .
\end{equation}
The results derived here for transitions among networks of size $N$ then apply to this model.


\begin{thebibliography}{60}
	\providecommand{\natexlab}[1]{#1}
	\providecommand{\url}[1]{\texttt{#1}}
	\expandafter\ifx\csname urlstyle\endcsname\relax
	\providecommand{\doi}[1]{doi: #1}\else
	\providecommand{\doi}{doi: \begingroup \urlstyle{rm}\Url}\fi
	
	\bibitem[Ohtsuki et~al.(2006)Ohtsuki, Hauert, Lieberman, and
	Nowak]{2006-Ohtsuki-p502-505}
	H.~Ohtsuki, C.~Hauert, E.~Lieberman, and M.~A. Nowak.
	\newblock {A simple rule for the evolution of cooperation on graphs and social
		networks}.
	\newblock \emph{Nature}, 441\penalty0 (7092):\penalty0 502--505, 2006.
	\newblock ISSN 14764687.
	\newblock \doi{10.1038/nature04605}.
	
	\bibitem[Holme and Saram{\"{a}}ki(2012)]{2012-Holme-p97-125}
	P.~Holme and J.~Saram{\"{a}}ki.
	\newblock {Temporal networks}.
	\newblock \emph{Physics Reports}, 519\penalty0 (3):\penalty0 97--125, 2012.
	\newblock ISSN 03701573.
	\newblock \doi{10.1016/j.physrep.2012.03.001}.
	
	\bibitem[Vazquez et~al.(2007)Vazquez, R{\'{a}}cz, Luk{\'{a}}cs, and
	Barab{\'{a}}si]{Vazquez-2007-prl}
	A.~Vazquez, B.~R{\'{a}}cz, A.~Luk{\'{a}}cs, and A.~L. Barab{\'{a}}si.
	\newblock {Impact of non-poissonian activity patterns on spreading processes}.
	\newblock \emph{Physical Review Letters}, 98:\penalty0 158702, 2007.
	\newblock ISSN 00319007.
	\newblock \doi{10.1103/PhysRevLett.98.158702}.
	
	\bibitem[Onnela et~al.(2007)Onnela, Saram{\"{a}}ki, Hyv{\"{o}}nen, Szab{\'{o}},
	Lazer, Kaski, Kert{\'{e}}sz, and Barab{\'{a}}si]{2007-Onnela-p7332-7336}
	J.~P. Onnela, J.~Saram{\"{a}}ki, J.~Hyv{\"{o}}nen, G.~Szab{\'{o}}, D.~Lazer,
	K.~Kaski, J.~Kert{\'{e}}sz, and A.~L. Barab{\'{a}}si.
	\newblock {Structure and tie strengths in mobile communication networks}.
	\newblock \emph{Proceedings of the National Academy of Sciences of the United
		States of America}, 104\penalty0 (18):\penalty0 7332--7336, 2007.
	\newblock ISSN 00278424.
	\newblock \doi{10.1073/pnas.0610245104}.
	
	\bibitem[Isella et~al.(2011)Isella, Stehl{\'{e}}, Barrat, Cattuto, Pinton, and
	{Van den Broeck}]{Isella-2011-jtb}
	L.~Isella, J.~Stehl{\'{e}}, A.~Barrat, C.~Cattuto, J.-F. Pinton, and W.~{Van
		den Broeck}.
	\newblock {What's in a crowd? Analysis of face-to-face behavioral networks}.
	\newblock \emph{Journal of Theoretical Biology}, 271\penalty0 (1):\penalty0
	166--180, 2011.
	\newblock ISSN 0022-5193.
	\newblock \doi{https://doi.org/10.1016/j.jtbi.2010.11.033}.
	
	\bibitem[Mastrandrea et~al.(2015)Mastrandrea, Fournet, and
	Barrat]{Mastrandrea2015}
	R.~Mastrandrea, J.~Fournet, and A.~Barrat.
	\newblock {Contact patterns in a high school: A comparison between data
		collected using wearable sensors, contact diaries and friendship surveys}.
	\newblock \emph{PLoS ONE}, 10\penalty0 (9):\penalty0 1--26, 2015.
	\newblock ISSN 19326203.
	\newblock \doi{10.1371/journal.pone.0136497}.
	
	\bibitem[Ulanowicz(2004)]{Ulanowicz-2004-cbc}
	R.~E. Ulanowicz.
	\newblock {Quantitative methods for ecological network analysis}.
	\newblock \emph{Computational Biology and Chemistry}, 28:\penalty0 321--339,
	2004.
	\newblock ISSN 14769271.
	\newblock \doi{10.1016/j.compbiolchem.2004.09.001}.
	
	\bibitem[Bascompte and Jordano(2007)]{Bascompte-2007-AREES}
	J.~Bascompte and P.~Jordano.
	\newblock {Plant-animal mutualistic networks: The architecture of
		biodiversity}.
	\newblock \emph{Annual Review of Ecology, Evolution, and Systematics},
	38:\penalty0 567--593, 2007.
	\newblock ISSN 1543592X.
	\newblock \doi{10.1146/annurev.ecolsys.38.091206.095818}.
	
	\bibitem[Miele and Matias(2017)]{Miele-2017-psos}
	V.~Miele and C.~Matias.
	\newblock {Revealing the hidden structure of dynamic ecological networks}.
	\newblock \emph{Royal Society Open Science}, 4\penalty0 (6):\penalty0 170251,
	2017.
	\newblock ISSN 20545703.
	\newblock \doi{10.1098/rsos.170251}.
	
	\bibitem[Akbarpour and Jackson(2018)]{Akbarpour-2018-pnas}
	M.~Akbarpour and M.~O. Jackson.
	\newblock {Diffusion in networks and the virtue of burstiness}.
	\newblock \emph{Proceedings of the National Academy of Sciences of the United
		States of America}, 115\penalty0 (30):\penalty0 E6996--E7004, 2018.
	\newblock ISSN 10916490.
	\newblock \doi{10.1073/pnas.1722089115}.
	
	\bibitem[Onaga et~al.(2017)Onaga, Gleeson, and Masuda]{Onaga-2017-prl}
	T.~Onaga, J.~P. Gleeson, and N.~Masuda.
	\newblock {Concurrency-induced transitions in epidemic dynamics on temporal
		networks}.
	\newblock \emph{Physical Review Letters}, 119\penalty0 (10):\penalty0 108301,
	2017.
	\newblock ISSN 10797114.
	\newblock \doi{10.1103/PhysRevLett.119.108301}.
	
	\bibitem[Kun and Scheuring(2009)]{2009-Kun-Biosystems}
	{\'{A}}.~Kun and I.~Scheuring.
	\newblock {Evolution of cooperation on dynamical graphs}.
	\newblock \emph{BioSystems}, 96\penalty0 (1):\penalty0 65--68, 2009.
	\newblock ISSN 03032647.
	\newblock \doi{10.1016/j.biosystems.2008.11.009}.
	
	\bibitem[Fulker et~al.(2021)Fulker, Forber, Smead, and Riedl]{Fulker2021}
	Z.~Fulker, P.~Forber, R.~Smead, and C.~Riedl.
	\newblock Spite is contagious in dynamic networks.
	\newblock \emph{Nature Communications}, 12\penalty0 (1), 2021.
	\newblock \doi{10.1038/s41467-020-20436-1}.
	
	\bibitem[Taylor et~al.(2007)Taylor, Day, and Wild]{taylor:Nature:2007}
	P.~D. Taylor, T.~Day, and G.~Wild.
	\newblock Evolution of cooperation in a finite homogeneous graph.
	\newblock \emph{Nature}, 447\penalty0 (7143):\penalty0 469--472, 2007.
	\newblock \doi{10.1038/nature05784}.
	
	\bibitem[Allen and McAvoy(2019)]{Allen2019}
	B.~Allen and A.~McAvoy.
	\newblock {A mathematical formalism for natural selection with arbitrary
		spatial and genetic structure}.
	\newblock \emph{Journal of Mathematical Biology}, 78\penalty0 (4):\penalty0
	1147--1210, 2019.
	\newblock ISSN 14321416.
	\newblock \doi{10.1007/s00285-018-1305-z}.
	
	\bibitem[Allen et~al.(2017)Allen, Lippner, Chen, Fotouhi, Momeni, Yau, and
	Nowak]{2017-Allen-p227-230}
	B.~Allen, G.~Lippner, Y.-T. Chen, B.~Fotouhi, N.~Momeni, S.-T. Yau, and M.~A.
	Nowak.
	\newblock {Evolutionary dynamics on any population structure}.
	\newblock \emph{Nature}, 544\penalty0 (7649):\penalty0 227--230, 2017.
	\newblock ISSN 0028-0836.
	
	\bibitem[Pacheco et~al.(2006)Pacheco, Traulsen, and
	Nowak]{2006-Pacheco-p258103}
	J.~M. Pacheco, A.~Traulsen, and M.~A. Nowak.
	\newblock {Coevolution of strategy and structure in complex networks with
		dynamical linking}.
	\newblock \emph{Physical Review Letters}, 97\penalty0 (25):\penalty0 258103,
	2006.
	\newblock ISSN 00319007.
	\newblock \doi{10.1103/PhysRevLett.97.258103}.
	
	\bibitem[Santos et~al.(2006)Santos, Pacheco, and
	Lenaerts]{2006-Santos-p1284-1291}
	F.~C. Santos, J.~M. Pacheco, and T.~Lenaerts.
	\newblock {Cooperation prevails when individuals adjust their social ties}.
	\newblock \emph{PLoS Computational Biology}, 2\penalty0 (10):\penalty0
	1284--1291, oct 2006.
	\newblock ISSN 1553734X.
	\newblock \doi{10.1371/journal.pcbi.0020140}.
	
	\bibitem[Pacheco et~al.(2008)Pacheco, Traulsen, Ohtsuki, and
	Nowak]{2008-pacheco-jtb}
	J.~M. Pacheco, A.~Traulsen, H.~Ohtsuki, and M.~A. Nowak.
	\newblock {Repeated games and direct reciprocity under active linking}.
	\newblock \emph{Journal of Theoretical Biology}, pages 723--731, 2008.
	\newblock ISSN 00225193.
	\newblock \doi{10.1016/j.jtbi.2007.10.040}.
	
	\bibitem[{Van Segbroeck} et~al.(2009){Van Segbroeck}, Santos, Lenaerts, and
	Pacheco]{2009-VanSegbroeck-p-}
	S.~{Van Segbroeck}, F.~C. Santos, T.~Lenaerts, and J.~M. Pacheco.
	\newblock {Reacting differently to adverse ties promotes cooperation in social
		networks}.
	\newblock \emph{Physical Review Letters}, 102\penalty0 (5):\penalty0 058105,
	feb 2009.
	\newblock ISSN 00319007.
	\newblock \doi{10.1103/PhysRevLett.102.058105}.
	
	\bibitem[Wu et~al.(2010)Wu, Zhou, Fu, Luo, Wang, and
	Traulsen]{2010-Wu-p11187-11187}
	B.~Wu, D.~Zhou, F.~Fu, Q.~Luo, L.~Wang, and A.~Traulsen.
	\newblock {Evolution of cooperation on stochastic dynamical networks}.
	\newblock \emph{PLoS ONE}, 5\penalty0 (6):\penalty0 e11187, 2010.
	\newblock \doi{10.1371/journal.pone.0011187}.
	
	\bibitem[Fehl et~al.(2011)Fehl, van~der Post, and Semmann]{2011-Fehl-p546-551}
	K.~Fehl, D.~J. van~der Post, and D.~Semmann.
	\newblock {Co-evolution of behaviour and social network structure promotes
		human cooperation}.
	\newblock \emph{Ecology Letters}, 14\penalty0 (6):\penalty0 546--551, 2011.
	\newblock ISSN 14610248.
	\newblock \doi{10.1111/j.1461-0248.2011.01615.x}.
	
	\bibitem[Rand et~al.(2011)Rand, Arbesman, and
	Christakis]{2011-Rand-p19193-19198}
	D.~G. Rand, S.~Arbesman, and N.~A. Christakis.
	\newblock {Dynamic social networks promote cooperation in experiments with
		humans}.
	\newblock \emph{Proceedings of the National academy of Sciences of the United
		States of America}, 108\penalty0 (48):\penalty0 19193--19198, 2011.
	\newblock \doi{10.1073/pnas.1108243108}.
	
	\bibitem[Bravo et~al.(2012)Bravo, Squazzoni, and Boero]{Bravo2012}
	G.~Bravo, F.~Squazzoni, and R.~Boero.
	\newblock {Trust and partner selection in social networks: An experimentally
		grounded model}.
	\newblock \emph{Social Networks}, 34:\penalty0 481--492, 2012.
	\newblock ISSN 03788733.
	\newblock \doi{10.1016/j.socnet.2012.03.001}.
	
	\bibitem[Wang et~al.(2012)Wang, Suri, and Watts]{2012-Wang-p14363-14368}
	J.~Wang, S.~Suri, and D.~J. Watts.
	\newblock {Cooperation and assortativity with dynamic partner updating}.
	\newblock \emph{Proceedings of the National Academy of Sciences of the United
		States of America}, 109\penalty0 (36):\penalty0 14363--14368, 2012.
	\newblock ISSN 00278424.
	\newblock \doi{10.1073/pnas.1120867109}.
	
	\bibitem[Bednarik et~al.(2014)Bednarik, Fehl, and Semmann]{Bednarik2014}
	P.~Bednarik, K.~Fehl, and D.~Semmann.
	\newblock {Costs for switching partners reduce network dynamics but not
		cooperative behaviour}.
	\newblock \emph{Proceedings of the Royal Society B: Biological Sciences},
	281:\penalty0 20141661, 2014.
	\newblock ISSN 14712954.
	\newblock \doi{10.1098/rspb.2014.1661}.
	
	\bibitem[Cardillo et~al.(2014)Cardillo, Petri, Nicosia, Sinatra,
	G{\'{o}}mez-Garde{\~{n}}es, and Latora]{2014-Cardillo-p52825-52825}
	A.~Cardillo, G.~Petri, V.~Nicosia, R.~Sinatra, J.~G{\'{o}}mez-Garde{\~{n}}es,
	and V.~Latora.
	\newblock {Evolutionary dynamics of time-resolved social interactions}.
	\newblock \emph{Physical Review E}, 90\penalty0 (5):\penalty0 52825, nov 2014.
	\newblock ISSN 15502376.
	\newblock \doi{10.1103/PhysRevE.90.052825}.
	
	\bibitem[Harrell et~al.(2018)Harrell, Melamed, and Simpson]{Harrell2018}
	A.~Harrell, D.~Melamed, and B.~Simpson.
	\newblock {The strength of dynamic ties: The ability to alter some ties
		promotes cooperation in those that cannot be altered}.
	\newblock \emph{Science Advances}, 4:\penalty0 eaau9109, 2018.
	\newblock ISSN 23752548.
	\newblock \doi{10.1126/sciadv.aau9109}.
	
	\bibitem[Ak{\c{c}}ay(2018)]{2018-Akcay-p2692-2692}
	E.~Ak{\c{c}}ay.
	\newblock {Collapse and rescue of cooperation in evolving dynamic networks}.
	\newblock \emph{Nature Communications}, 9\penalty0 (1):\penalty0 2692, jul
	2018.
	\newblock ISSN 20411723.
	\newblock \doi{10.1038/s41467-018-05130-7}.
	
	\bibitem[Wong and Candolin(2014)]{wong:BE:2014}
	B.~B.~M. Wong and U.~Candolin.
	\newblock Behavioral responses to changing environments.
	\newblock \emph{Behavioral Ecology}, 26\penalty0 (3):\penalty0 665--673, 2014.
	\newblock \doi{10.1093/beheco/aru183}.
	
	\bibitem[Tilman et~al.(2020)Tilman, Plotkin, and Ak{\c{c}}ay]{tilman:NC:2020}
	A.~R. Tilman, J.~B. Plotkin, and E.~Ak{\c{c}}ay.
	\newblock Evolutionary games with environmental feedbacks.
	\newblock \emph{Nature Communications}, 11\penalty0 (1):\penalty0 915, 2020.
	\newblock \doi{10.1038/s41467-020-14531-6}.
	
	\bibitem[Nowak et~al.(2004)Nowak, Sasaki, Taylor, and
	Fudenberg]{2004-Nowak-p646-650}
	M.~A. Nowak, A.~Sasaki, C.~Taylor, and D.~Fudenberg.
	\newblock {Emergence of cooperation and evolutionary stability in finite
		populations}.
	\newblock \emph{Nature}, 428\penalty0 (6983):\penalty0 646--650, 2004.
	\newblock ISSN 00280836.
	\newblock \doi{10.1038/nature02414}.
	
	\bibitem[McAvoy and Allen(2021)]{McAvoy2021}
	A.~McAvoy and B.~Allen.
	\newblock {Fixation probabilities in evolutionary dynamics under weak
		selection}.
	\newblock \emph{Journal of Mathematical Biology}, 82\penalty0 (3):\penalty0 14,
	2021.
	\newblock ISSN 1432-1416.
	\newblock \doi{10.1007/s00285-021-01568-4}.
	
	\bibitem[Fisher(1930)]{fisher:OUP:1930}
	R.~A. Fisher.
	\newblock \emph{{The Genetical Theory of Natural Selection}}.
	\newblock Clarendon Press, 1930.
	\newblock \doi{10.5962/bhl.title.27468}.
	
	\bibitem[Taylor(1990)]{taylor:AN:1990}
	P.~D. Taylor.
	\newblock {Allele-frequency change in a class-structured population}.
	\newblock \emph{The American Naturalist}, 135\penalty0 (1):\penalty0 95--106,
	1990.
	\newblock \doi{10.1086/285034}.
	
	\bibitem[McAvoy and Wakeley(2022)]{mcavoy:PNAS:2022}
	A.~McAvoy and J.~Wakeley.
	\newblock Evaluating the structure-coefficient theorem of evolutionary game
	theory.
	\newblock \emph{Proceedings of the National Academy of Sciences of the United
		States of America}, 119\penalty0 (28):\penalty0 e2119656119, 2022.
	\newblock \doi{10.1073/pnas.2119656119}.
	
	\bibitem[Erd{\"{o}}s and R{\'{e}}nyi(1960)]{1960-Erdoes-p17-61}
	P.~Erd{\"{o}}s and A~R{\'{e}}nyi.
	\newblock {On the evolution of random graphs}.
	\newblock \emph{Publication of the Mathematical Institute of the Hungarian
		Academy of Sciences}, 5:\penalty0 17--61, 1960.
	
	\bibitem[Goh et~al.(2001)Goh, Kahng, and Kim]{2001-Goh-p278701-278701}
	K.~I. Goh, B.~Kahng, and D.~Kim.
	\newblock {Universal behavior of load distribution in scale-free networks}.
	\newblock \emph{Physical Review Letters}, 87\penalty0 (27):\penalty0 278701,
	2001.
	\newblock ISSN 10797114.
	\newblock \doi{10.1103/PhysRevLett.87.278701}.
	
	\bibitem[Gernat et~al.(2018)Gernat, Rao, Middendorf, Dankowicz, Goldenfeld, and
	Robinson]{gernat:PNAS:2018}
	T.~Gernat, V.~D. Rao, M.~Middendorf, H.~Dankowicz, N.~Goldenfeld, and G.~E.
	Robinson.
	\newblock Automated monitoring of behavior reveals bursty interaction patterns
	and rapid spreading dynamics in honeybee social networks.
	\newblock \emph{Proceedings of the National Academy of Sciences of the United
		States of America}, 115\penalty0 (7):\penalty0 1433--1438, 2018.
	\newblock \doi{10.1073/pnas.1713568115}.
	
	\bibitem[Watts and Strogatz(1998)]{1998-Watts-p440-442}
	D.~J. Watts and S.~H. Strogatz.
	\newblock {Collective dynamics of `small-world' networks}.
	\newblock \emph{Nature}, 393\penalty0 (6684):\penalty0 440--442, 1998.
	
	\bibitem[Barab{\'{a}}si and Albert(1999)]{1999-Barabasi-p509-512}
	A.-L. Barab{\'{a}}si and R.~Albert.
	\newblock {Emergence of scaling in random networks}.
	\newblock \emph{Science}, 286\penalty0 (5439):\penalty0 509--512, oct 1999.
	
	\bibitem[Holme and Kim(2002)]{Holme2002}
	P.~Holme and B.~J. Kim.
	\newblock {Growing scale-free networks with tunable clustering}.
	\newblock \emph{Physical Review E}, 65\penalty0 (2):\penalty0 2--5, 2002.
	\newblock ISSN 1063651X.
	\newblock \doi{10.1103/PhysRevE.65.026107}.
	
	\bibitem[Lieberman et~al.(2005)Lieberman, Hauert, and
	Nowak]{2005-Lieberman-p312-316}
	E.~Lieberman, C.~Hauert, and M.~A. Nowak.
	\newblock {Evolutionary dynamics on graphs}.
	\newblock \emph{Nature}, 433\penalty0 (7023):\penalty0 312--316, jan 2005.
	\newblock ISSN 0028-0836.
	
	\bibitem[Pan and Saram{\"{a}}ki(2011)]{Pan-2011-pre}
	R.~K. Pan and J.~Saram{\"{a}}ki.
	\newblock {Path lengths, correlations, and centrality in temporal networks}.
	\newblock \emph{Physical Review E}, 84:\penalty0 016105, 2011.
	\newblock ISSN 15393755.
	\newblock \doi{10.1103/PhysRevE.84.016105}.
	
	\bibitem[Holme(2015)]{2015-Holme-p234-234}
	P.~Holme.
	\newblock {Modern temporal network theory: a colloquium}.
	\newblock \emph{European Physical Journal B}, 88\penalty0 (9):\penalty0 234,
	2015.
	\newblock ISSN 14346036.
	\newblock \doi{10.1140/epjb/e2015-60657-4}.
	
	\bibitem[Valdano et~al.(2015)Valdano, Ferreri, Poletto, and
	Colizza]{2015-Valdano-p21005-21005}
	E.~Valdano, L.~Ferreri, C.~Poletto, and V.~Colizza.
	\newblock {Analytical computation of the epidemic threshold on temporal
		networks}.
	\newblock \emph{Physical Review X}, 5\penalty0 (2):\penalty0 21005, apr 2015.
	\newblock \doi{10.1103/PhysRevX.5.021005}.
	
	\bibitem[Povinelli et~al.(1992)Povinelli, Nelson, and
	Boysen]{povinelli:AB:1992}
	D.~J. Povinelli, K.~E. Nelson, and S.~T. Boysen.
	\newblock Comprehension of role reversal in chimpanzees: evidence of empathy?
	\newblock \emph{Animal Behaviour}, 43\penalty0 (4):\penalty0 633--640, 1992.
	\newblock \doi{10.1016/s0003-3472(05)81022-x}.
	
	\bibitem[Su et~al.(2022)Su, Allen, and Plotkin]{su-2022-pnas}
	Q.~Su, B.~Allen, and J.~B. Plotkin.
	\newblock {Evolution of cooperation with asymmetric social interactions}.
	\newblock \emph{Proceedings of the National Academy of Sciences of the United
		States of America}, 119\penalty0 (1):\penalty0 e2113468118, jan 2022.
	\newblock \doi{10.1073/pnas.2113468118}.
	
	\bibitem[Peysakhovich et~al.(2014)Peysakhovich, Nowak, and
	Rand]{peysakhovich-nc-2014}
	A.~Peysakhovich, M.~A. Nowak, and D.~G. Rand.
	\newblock {Humans display a 'cooperative phenotype' that is domain general and
		temporally stable}.
	\newblock \emph{Nature Communications}, 5:\penalty0 4939, 2014.
	\newblock ISSN 20411723.
	\newblock \doi{10.1038/ncomms5939}.
	
	\bibitem[Harmer and Abbott(1999)]{harmer:Nature:1999}
	G.~Harmer and D.~Abbott.
	\newblock {Losing strategies can win by Parrondo's paradox}.
	\newblock \emph{Nature}, 402:\penalty0 864, 1999.
	\newblock \doi{10.1038/47220}.
	
	\bibitem[Girvan and Newman(2002)]{girvan-2002-pnas}
	M.~Girvan and M.~E.J. Newman.
	\newblock {Community structure in social and biological networks}.
	\newblock \emph{Proceedings of the National Academy of Sciences of the United
		States of America}, 99\penalty0 (12):\penalty0 7821--7826, 2002.
	\newblock ISSN 00278424.
	\newblock \doi{10.1073/pnas.122653799}.
	
	\bibitem[Newman(2006)]{newman-2006-pnas}
	M.~E.J. Newman.
	\newblock {Modularity and community structure in networks}.
	\newblock \emph{Proceedings of the National Academy of Sciences of the United
		States of America}, 103\penalty0 (23):\penalty0 8577--8582, 2006.
	\newblock ISSN 00278424.
	\newblock \doi{10.1073/pnas.0601602103}.
	
	\bibitem[Blonder et~al.(2012)Blonder, Wey, Dornhaus, James, and
	Sih]{Blonder-2012-mee}
	B.~Blonder, T.~W. Wey, A.~Dornhaus, R.~James, and A.~Sih.
	\newblock {Temporal dynamics and network analysis}.
	\newblock \emph{Methods in Ecology and Evolution}, 3:\penalty0 958--972, 2012.
	\newblock ISSN 2041210X.
	\newblock \doi{10.1111/j.2041-210X.2012.00236.x}.
	
	\bibitem[Fudenberg and Imhof(2006)]{fudenberg:JET:2006}
	D.~Fudenberg and L.~A. Imhof.
	\newblock Imitation processes with small mutations.
	\newblock \emph{Journal of Economic Theory}, 131\penalty0 (1):\penalty0
	251--262, 2006.
	\newblock \doi{10.1016/j.jet.2005.04.006}.
	
	\bibitem[Tarnita et~al.(2009{\natexlab{a}})Tarnita, Ohtsuki, Antal, Fu, and
	Nowak]{tarnita:JTB:2009}
	C.~E. Tarnita, H.~Ohtsuki, T.~Antal, F.~Fu, and M.~A. Nowak.
	\newblock Strategy selection in structured populations.
	\newblock \emph{Journal of Theoretical Biology}, 259\penalty0 (3):\penalty0
	570--581, 2009{\natexlab{a}}.
	\newblock \doi{10.1016/j.jtbi.2009.03.035}.
	
	\bibitem[Tarnita et~al.(2009{\natexlab{b}})Tarnita, Antal, Ohtsuki, and
	Nowak]{tarnita:PNAS:2009}
	C.~E. Tarnita, T.~Antal, H.~Ohtsuki, and M.~A. Nowak.
	\newblock Evolutionary dynamics in set structured populations.
	\newblock \emph{Proceedings of the National Academy of Sciences of the United
		States of America}, 106\penalty0 (21):\penalty0 8601--8604,
	2009{\natexlab{b}}.
	\newblock \doi{10.1073/pnas.0903019106}.
	
	\bibitem[Tarnita et~al.(2011)Tarnita, Wage, and Nowak]{tarnita:PNAS:2011}
	C.~E. Tarnita, N.~Wage, and M.~A. Nowak.
	\newblock Multiple strategies in structured populations.
	\newblock \emph{Proceedings of the National Academy of Sciences of the United
		States of America}, 108\penalty0 (6):\penalty0 2334--2337, 2011.
	\newblock \doi{10.1073/pnas.1016008108}.
	
	\bibitem[D{\'{e}}barre(2019)]{debarre:DGA:2019}
	F.~D{\'{e}}barre.
	\newblock {Imperfect Strategy Transmission Can Reverse the Role of Population
		Viscosity on the Evolution of Altruism}.
	\newblock \emph{Dynamic Games and Applications}, 10\penalty0 (3):\penalty0
	732--763, 2019.
	\newblock \doi{10.1007/s13235-019-00326-y}.
	
	\bibitem[McAvoy et~al.(2020)McAvoy, Allen, and Nowak]{mcavoy-2020-nhb}
	A.~McAvoy, B.~Allen, and M.~A. Nowak.
	\newblock {Social goods dilemmas in heterogeneous societies}.
	\newblock \emph{Nature Human Behaviour}, 4\penalty0 (8):\penalty0 819--831, aug
	2020.
	\newblock ISSN 23973374.
	\newblock \doi{10.1038/s41562-020-0881-2}.
	
	\bibitem[Allen and McAvoy(2022)]{allen:preprint:2022}
	B.~Allen and A.~McAvoy.
	\newblock The coalescent with arbitrary spatial and genetic structure.
	\newblock arXiv preprint arXiv:2207.02880, 2022.
	
\end{thebibliography}
\end{document}